\documentclass[11pt]{article}
\pdfoutput=1

\usepackage[margin=1in]{geometry}
\usepackage[utf8]{inputenc}
\usepackage[fleqn]{amsmath}
\usepackage{amssymb, amsthm}
\usepackage{cite}
\usepackage{hyperref}
\usepackage{authblk}
\usepackage{graphicx}
\usepackage[labelformat=simple]{subcaption}

\usepackage[ruled, vlined, linesnumbered]{algorithm2e}
\usepackage{booktabs}
\usepackage[export]{adjustbox}
\usepackage{microtype}
\usepackage{tikz}
\usetikzlibrary{quantikz}

\usepackage{xcolor}

\newcommand{\bs}[1]{\boldsymbol #1}
\newtheorem{definition}{Definition}
\newtheorem{lemma}{Lemma}
\newtheorem{theorem}{Theorem}
\newtheorem{problem}{Problem}
\DeclareMathOperator{\rank}{rank}
\DeclareMathOperator{\param}{par}

\title{Optimal number of parametrized rotations and Hadamard gates in parametrized Clifford circuits with non-repeated parameters}

\author[1,2]{Vivien Vandaele}
\author[2]{Simon Perdrix}
\author[2]{Christophe Vuillot}

\affil[1]{Eviden Quantum Lab, Les Clayes-sous-Bois, France}
\affil[2]{Université de Lorraine, CNRS, Inria, LORIA, F-54000 Nancy, France}
\date{}

\begin{document}
\setbool{@fleqn}{false}
\maketitle

\begin{abstract}
    We present an efficient algorithm to reduce the number of non-Clifford gates in quantum circuits and the number of parametrized rotations in parametrized quantum circuits.
    The method consists in finding rotations that can be merged into a single rotation gate.
    This approach has already been considered before and is used as a pre-processing procedure in many optimization algorithms, notably for optimizing the number of Hadamard gates or the number of $T$ gates in Clifford$+T$ circuits.
    Our algorithm has a better complexity than similar methods and is particularly efficient for circuits with a low number of internal Hadamard gates.
    Furthermore, we show that this approach is optimal for parametrized circuits composed of Clifford gates and parametrized rotations with non-repeated parameters.
    For the same type of parametrized quantum circuits, we also prove that a previous procedure optimizing the number of Hadamard gates and internal Hadamard gates is optimal.
    This procedure is notably used in our low-complexity algorithm for optimally reducing the number of parametrized rotations.
\end{abstract}

\section{Introduction}

Non-Clifford rotation gates are typically more resource-intensive to implement fault-tolerantly than Clifford gates.
Finding efficient strategies to optimize the number of these gates is therefore crucial to reduce the overall fault-tolerant implementation cost of a quantum circuit.
That is why various methods have been developed to optimize the number of non-Clifford rotation gates, notably for Clifford$+R_Z$ circuits~\cite{zhang2019optimizing, kissinger2020reducing, heyfron2018efficient, de2020fast, ruiz2024quantum}.
Optimizing the number of non-Clifford rotation gates is also particularly important in parametrized quantum circuits.
Parametrized quantum circuits are quantum circuits that include adjustable parameters in their quantum gates.
They are central to variational quantum algorithms such as QAOA~\cite{farhi2014quantum} or VQE~\cite{peruzzo2014variational}.
In such circuits, the parameterized rotation gates are often non-Clifford gates, making their optimization crucial for improving the efficiency of these algorithms.

A simple yet efficient strategy to reduce the number of non-Clifford rotation gates in a quantum circuit is to find rotations that can be merged into a single rotation gate~\cite{zhang2019optimizing, kissinger2020reducing}.
To find the non-Clifford rotation gates which can be merged, it is useful to represent the operation performed by the circuit as a sequence of Pauli rotations followed by a final Clifford operator as done in Reference~\cite{zhang2019optimizing}.
The number of Pauli rotations in the sequence is then equal to the number of non-Clifford rotation gates in the circuit.
If the Pauli products of two Pauli rotations in the sequence are identical and if these Pauli rotations are not separated by another Pauli rotation in the sequence with which they anticommute, then they can be merged into a single Pauli rotation.
The complexity of this method, first developed in Reference~\cite{zhang2019optimizing}, mainly resides in the number of commutativity checks (i.e.\ checking whether or not two Pauli rotations commute) that must be performed.
In the worst case, the commutativity between each pair of Pauli rotations must be evaluated, leading to $\mathcal{O}(m^2)$ commutativity checks, where $m$ is the number of Pauli rotations.
In this work, we show how the number of commutativity checks can be greatly reduced.
More precisely, we present a method in which the number of commutativity checks is $\mathcal{O}(mh)$ in the worst case, where $h$ satisfies $h < m$ and is the minimal number of internal Hadamard gates required to implement the sequence of Pauli rotations.
The number of internal Hadamard gates corresponds to the number of Hadamard gates lying between the first and the last non-Clifford rotation gate of the circuit.
An algorithm to implement a sequence of Pauli rotations with a minimal number of internal Hadamard gates was presented in Reference~\cite{vandaele2024optimal}.
This algorithm is used in our method to reduce the number of commutativity checks.

We prove that our algorithm finds all the Pauli rotations which can be merged into a single Pauli rotation when the angles of the Pauli rotations are black boxes.
However, the algorithm sometimes falls short of achieving the same reduction in the number of $T$ gates as analogous algorithms for Clifford$+T$ circuits.
We propose a solution to this shortcoming which leads to an algorithm with an increased worst-case complexity but that is faster than analogous algorithms and that achieves the same reduction in the number of $T$ gates.

An interesting property of this approach is that the structure of the quantum circuit is preserved.
Only non-Clifford gates are modified: some are removed, while the angles of others are adjusted.
As such, this approach reduces the number of non-Clifford gates in the circuit without increasing the number of Clifford gates.
That is why this approach is used to reduce the number of non-Clifford gates in various quantum compilers designed for NISQ devices~\cite{amy2020staq, sivarajah2020t, martiel2022architecture}.
This way of optimizing the number of $T$ gates is outperformed by some other approaches on some quantum circuits.
However, this strategy generally leads to an important reduction in the number of $T$ gates and its execution is much faster than other approaches.
Although this approach does not affect the Clifford gates in the circuit, the reduction of the number of $T$ gates can improve the performances of other algorithms designed for optimizing the number of Clifford gates.
For example, it has been shown that the algorithm presented in Reference~\cite{vandaele2024optimal} for reducing the number of internal Hadamard gates achieves much better reductions when the number of $T$ gates have been optimized using this strategy.
Reducing the number of internal Hadamard gates is an important pre-processing step to realize before executing some other $T$-count optimizers, as fewer internal Hadamard gates result in shorter execution times, fewer $T$ gates, and a reduced need for ancillary qubits in the optimized circuit~\cite{heyfron2018efficient, de2020fast, ruiz2024quantum}.

Our low-complexity algorithm can also be used to optimize the number of parametrized rotation gates in parametrized quantum circuits.
We prove that our algorithm finds the minimal number of parametrized rotation gates when the circuit is composed of Clifford gates and parametrized rotation gates, and when each parameter can only appear once in the circuit.
Our results mainly differ from Reference~\cite{van2024optimal} in the parameter transformations considered between equivalent parametrized quantum circuits.
The equivalence between parametrized quantum circuits considered in Reference~\cite{van2024optimal} is restricted to cases where the relations between parameters are analytic functions.
In this work, we do not impose such a restriction on the relations between parameters.
A consequence of the restricted parameter transformations considered in Reference~\cite{van2024optimal} is that the following Euler angle transformation rewriting rule:
\begin{equation}
    \includegraphics[valign=c]{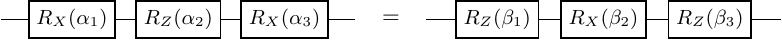}
\end{equation}
where $\beta_1, \beta_2, \beta_3$ are dependent on $\alpha_1, \alpha_2, \alpha_3$ through some trigonometric relations, is excluded as it involves discontinuous parameter transformations.
Our results show that this rewriting rule is not useful for finding an equivalent parametrized Clifford circuit with an optimal number of parametrized rotation gates and non-repeated parameters.

On the basis of this result, we also prove that the algorithm presented in Reference~\cite{vandaele2024optimal} to optimize the number of Hadamard gates and internal Hadamard gates is optimal for parametrized Clifford circuits with non-repeated parameters.
For parametrized circuits composed of Clifford gates and parametrized rotation gates, the number of internal Hadamard gates corresponds to the number of Hadamard gate lying between the first and the last parametrized rotation gate of the circuit.
Combined with our algorithm designed to optimize the number of parametrized rotation gates, this leads to a procedure for achieving an optimal number of parametrized rotation gates, Hadamard gates, and internal Hadamard gates, with complexity of $\mathcal{O}(nM + nhm)$, where $n$ is the number of qubits, $M$ is the number of gates in the initial circuit, $m$ is the initial number of parametrized rotation gates and $h$ is the optimal number of internal Hadamard gates.

In Section~\ref{sec:preliminaries}, we introduce the necessary notions and definitions to present our results.
Then, in Section~\ref{sec:merging}, we present an efficient algorithm for optimizing the number of non-Clifford rotation gates and for optimizing the number of parametrized rotation gates in parametrized quantum circuits.
The optimality of our algorithm for parametrized Clifford circuits with non-repeated parameters is proven in Section~\ref{sec:opt_parametrized}.
On the basis of this result, in Section~\ref{sec:opt_parametrized_h}, we prove that the Hadamard gate optimization algorithm presented in Reference~\cite{vandaele2024optimal} is optimal in the number of Hadamard gates and internal Hadamard gates for parametrized Clifford circuits with non-repeated parameters.
Finally, we provide benchmarks in Section~\ref{sec:bench} to evaluate the performances of our algorithm on a library of reversible logic circuits.

\section{Preliminaries}\label{sec:preliminaries}

\subsection{Pauli rotation}\label{sub:pauli_rotation}

We define the set of Pauli operators $\mathcal{P}_n$ as the set composed of all the tensor products of $n$ Pauli matrices, which are defined as follows:
$$
I = \begin{pmatrix}
1 & 0\\
0 & 1
\end{pmatrix},\quad
X = \begin{pmatrix}
0 & 1\\
1 & 0
\end{pmatrix},\quad
Y = \begin{pmatrix}
0 & -i\\
i & 0
\end{pmatrix},\quad
Z = \begin{pmatrix}
1 & 0\\
0 & -1
\end{pmatrix},
$$
with a mutiplicative factor of $\pm 1$.
We will use the term Pauli product to designate a Pauli operator deprived of its sign.
A Pauli rotation $R_P(\theta)$ is defined as follows:
$$R_P(\theta) = \exp(-i\theta P/2) = \cos(\theta/2)I - i\sin(\theta/2)P$$
for a Pauli operator $P \in \mathcal{P}_n$ and an angle $\theta \in \mathbb{R}$.
For instance, the $T$ gate is defined as a $\pi/4$ Pauli $Z$ rotation: 
$$T = R_Z(\pi/4).$$
Clifford gates can also be represented in terms of Pauli rotations.
For exemple, the CNOT gate and the Hadamard gate (denoted by $H$) are defined as follows, up to a global phase:
\begin{align*}
    \mathrm{CNOT} &= R_{ZX}(\pi/2)R_{ZI}(-\pi/2)R_{IX}(-\pi/2), \\
    H &= R_{Z}(\pi/2) R_{X}(\pi/2) R_{Z}(\pi/2).
\end{align*}
The Clifford group, denoted $\mathcal{C}_n$, is generated by the set of $\pi/2$ Pauli rotations acting on $n$ qubits:
$$\{R_P(\pi/2) \mid P \in \mathcal{P}_n\}.$$
A Pauli operator $P \in \mathcal{P}_n$ conjugated by a Clifford gate $U \in \mathcal{C}_n$ is always equal to another Pauli operator $P' \in \mathcal{P}_n$, i.e.\ $U^\dag P U = P'$. 
This fact also holds when a Pauli rotation is conjugated by $U \in \mathcal{C}_n$: 
\begin{equation}
    U^\dag R_P(\theta) U = R_{U^\dag P U }(\theta) = R_{P'}(\theta).
\end{equation}
That is why a unitary gate $U$, representing the operation performed by a quantum circuit acting on $n$ qubits and composed of Clifford gates and single-qubit rotation gates, can always be described by a sequence of Pauli rotations and a final Clifford operator $C \in \mathcal{C}_n$~\cite{gosset2014algorithm}:
\begin{equation}\label{eq:sequence_prod}
    U = e^{i \phi} C \left( \prod_{i=1}^{m} R_{P_i}(\theta_i)\right)
\end{equation}
where $m$ is the number of non-Clifford rotation gates in the circuit and $P_i$ is a Pauli product of size $n$ different from $I^{\otimes n}$, for all $i$. Let $P \in \mathcal{P}_n$ and $P' \in \mathcal{P}_n$, $P$ commutes with $P'$ if and only if their associated Pauli products commute on an even number of elements, otherwise $P$ anticommutes with $P'$.
Also, two Pauli rotations $R_P(\theta)$ and $R_{P'}(\theta')$, where $\theta$ and $\theta'$ are satisfying $\theta \neq 0 \pmod{2\pi}$ and $\theta' \neq 0 \pmod{2\pi}$, commute if and only if $P$ commutes with $P'$.

Consider the sequence of $m$ Pauli products $P_1, \ldots, P_m$ represented in Equation~\ref{eq:sequence_prod}.
If this sequence contains two Pauli products $P_i$ and $P_j$ where $i,j$ are satisfying $1 \leq i < j \leq m$ and such that $P_i = P_j$ and if there exists no $k$ satisfying $i < k < j$ and such that $P_k$ anticommutes with $P_i$, then the two Pauli rotations associated with $P_i$ and $P_j$ can be merged into a single Pauli rotation.
Merging these two Pauli rotations will form a Pauli rotation equal to $R_{P_i}(\theta_i + \theta_j)$.
This way of modifying the sequence of Pauli rotations leads to a simple and effective method to reduce the number of Pauli rotations, and therefore the number of non-Clifford rotation gates. This method was first introduced in Reference~\cite{zhang2019optimizing}.
An analogous method was independently introduced in Reference~\cite{kissinger2020reducing}.
Although the approaches taken in these references are different, it has been proven that they are in fact equivalent~\cite{simmons2021relating}.
This method of optimizing the number of Pauli rotations leads to the following Pauli rotation merging problem.

\begin{problem}[Pauli rotation merging]\label{pb:rotation_merging}
    Given a sequence of Pauli rotations $R_{P_1}(\theta_1), \ldots, R_{P_m}(\theta_m)$, find all pairs of integers $i, j$ satisfying $1 \leq i < j \leq m$, $P_i = P_j$ and such that there exists no integer $k$ satisfying $i < k < j$ such that $P_i$ anticommutes with $P_k$.
\end{problem}

In Section~\ref{sec:merging}, we present an algorithm that efficiently solves the Pauli rotation merging problem and constructs the associated optimized quantum circuit.
Benchmarks are given in Section~\ref{sec:bench} to evaluate the performances of our algorithm on Clifford$+T$ circuits.

\subsection{Hadamard gate optimization}

A method for optimizing the number of Hadamard gates in a Clifford$+R_Z$ quantum circuit consists in describing the circuit in the form of Equation~\ref{eq:sequence_prod} and then synthesizing the sequence of Pauli rotations and the final Clifford operator with a minimal number of Hadamard gates.
An optimal algorithm for this approach is presented in Reference~\cite{vandaele2024optimal}.
The algorithm also performs the synthesis of the sequence of Pauli rotations with a minimal number of internal Hadamard gates, where the number of internal Hadamard gates corresponds to the number of Hadamard gates lying between the first and the last non-Clifford $R_Z$ gate of the circuit.
Moreover, the number of internal Hadamard gate in the circuit produced by the algorithm is equal to $\rank(A)$, where $A$ is the commutativity matrix associated with the initial circuit, defined as follows.

\begin{definition}[Commutativity matrix]
    Let $P_1, \ldots, P_m$ be a sequence of Pauli products. 
    The commutativity matrix $A$ associated with this sequence of Pauli products is a strictly upper triangular Boolean matrix of size $m \times m$ such that for all integers $i, j$ satisfying $0 \leq i < j < m $:
    $$
    A_{i, j} = 
    \begin{cases}
        0 &\text{if $P_i$ commutes with $P_j$},\\ 
        1 &\text{if $P_i$ anticommutes with $P_j$}.
    \end{cases}
    $$
\end{definition}

We will say that a commutativity matrix $A$ is associated with a Clifford$+R_Z$ circuit $C$ if and only if $A$ is associated with the sequence of Pauli products obtained when $C$ is put into the form of Equation~\ref{eq:sequence_prod}.
Similarly, the number of Hadamard gate in the circuit produced by the algorithm is equal to $\rank(M)$, where $M$ is the extended commutativity matrix associated with the initial circuit, defined as follows.

\begin{definition}[Extended commutativity matrix]
    Let $C$ be a Clifford$+R_Z$ circuit, and let $C'$ be a Clifford$+R_Z$ circuit constructed from $C$ by inserting an $R_Z(\theta)$ gate, with $\theta$ not being a multiple of $\pi/2$, at the beginning and at the end of the circuit on each qubit.
    The extended commutativity matrix $M$ associated with $C$ is the commutativity matrix associated with $C'$.
\end{definition}

In Section~\ref{sec:merging}, we use the Hadamard gate optimization algorithm presented in Reference~\cite{vandaele2024optimal} to design an efficient algorithm for optimizing the number of non-Clifford $R_Z$ gates in Clifford$+R_Z$ circuits.
Then, in Section~\ref{sec:opt_parametrized_h}, we prove that the algorithm presented in Reference~\cite{vandaele2024optimal} to optimize the number of Hadamard gates and internal Hadamard gates is optimal for parametrized Clifford circuits with non-repeated parameters.

\subsection{Parametrized Clifford circuit}

We are interested in Clifford$+R_Z$ parametrized quantum circuits that do not have non-Clifford constant rotation gates.
Such parametrized quantum circuits are called parametrized Clifford circuits and are defined as follows.

\begin{definition}[Parametrized Clifford circuit]
    A parametrized Clifford circuit with parameters $\bs \alpha \in \mathbb{R}^\kappa$ is a quantum circuit composed of Clifford gates and parametrized rotation gates $R_Z(f_1(\bs \alpha)),$ $\ldots,$ $R_Z(f_m(\bs \alpha))$, where $f_i$ is any non-constant function satisfying $f_i(\bs 0) = 0$ and $f_i(\bs \alpha) \in [0, 2\pi)$ for all $\bs \alpha$.
\end{definition}

Similarly to Equation~\ref{eq:sequence_prod}, the parametrized unitary gate $U$ associated with a parametrized Clifford circuit acting on $n$ qubits and composed of parametrized rotation gates $R_Z(f_1(\bs \alpha)), \ldots, R_Z(f_m(\bs \alpha))$ can be represented by a sequence of parametrized Pauli rotations and a Clifford operator:
\begin{equation}\label{eq:sequence_prod_parametrized}
    U(\bs \alpha) = e^{i\phi} C \left( \prod_{i=1}^{m} R_{P_i}(f_i(\bs \alpha))\right)
\end{equation}
where $P_i \in \mathcal{P}_n \setminus \{\pm I^{\otimes n}\}$, $C$ is a Clifford operator, and assuming that either $R_Z(f_i(\bs \alpha))$ appears after $R_Z(f_{i+1}(\bs \alpha))$ in the circuit, or that the two gates are parallel in the circuit.
Note that the Pauli rotations in Equation~\ref{eq:sequence_prod_parametrized} are not necessarily non-Clifford as we can have Clifford parametrized rotation gates in the circuit.
A sequence of Pauli rotations associated with a parametrized Clifford circuit refers to a sequence of Pauli rotations obtained when the circuit is put into the form of Equation~\ref{eq:sequence_prod_parametrized}.

\begin{definition}
Given $S\subseteq [1..\kappa]$, let $p_S:\mathbb R^\kappa \to \mathbb R^\kappa$ be the projector on $S$ defined as, $\forall {\bs \alpha}\in \mathbb R^\kappa, \forall i\in [1..\kappa]$,  
\begin{equation}
    (p_S({\bs \alpha}))_i = \begin{cases} 
        \bs \alpha_i& \text{if $i\in S$}, \\
        0 & \text{otherwise}. \end{cases}
\end{equation}
\end{definition}
As a function may not depend on some of its parameters, we define its useful parameters as the parameters that actually have an effect on the value of a function.
\begin{definition}[Useful parameters]
    The useful parameters of a function $f: \mathbb{R}^\kappa \rightarrow \mathbb{R}$ is the smallest\footnote{such set is unique as $f\circ p_S=f$ and $f\circ p_{S'}=f$ imply that $f\circ p_{S\cap S'}=f$.} set $\param(f)$ such that  $f\circ p_{\param(f)} = f$.
\end{definition}
Note that the functions in a parametrized Clifford circuit always have a non-empty set of useful parameters as they are non-constant.
We will also be interested in particular subsets of the useful parameters of a function, referred to as minimal supports and defined as follows.
\begin{definition}[Minimal support]
A \emph{minimal support} of a function $f: \mathbb{R}^\kappa \rightarrow \mathbb{R}$ is a set $S \subseteq [1..\kappa]$ such that $f\circ p_S\neq 0$ and $\forall T\subset S$, $f\circ p_T=0$.
\end{definition}
For example, the minimal supports of the function $f(\alpha_1, \alpha_2) = \alpha_1 + \alpha_2$ are $\{1\}$ and $\{2\}$.
Whereas the only minimal support of the function $f(\alpha_1, \alpha_2) = \alpha_1 \alpha_2$ is $\{1, 2\}$.
Note that a non-constant function always have at least one minimal support.
This notion of minimal support will be useful for the study of parametrized Clifford circuit with non-repeated parameters, defined as follows.

\begin{definition}[Parametrized Clifford circuit with non-repeated parameters]
A parametrized Clifford circuit with non-repeated parameters is a parametrized Clifford circuit with parameters $\bs \alpha \in \mathbb{R}^\kappa$ and parametrized rotation gates $R_Z(f_1(\bs \alpha)), \ldots, R_Z(f_m(\bs \alpha))$ where $f_i$ and $f_j$ have disjoint sets of useful parameters for all integers $i, j$ satisfying $1 \leq i < j \leq m$.
\end{definition}

In Section~\ref{sec:opt_parametrized}, we prove that the algorithm presented in Section~\ref{sec:merging} is optimally minimizing the number of parametrized rotation gates in parametrized Clifford circuit with non-repeated parameters.
Then, in Section~\ref{sec:opt_parametrized_h}, we prove that the algorithm presented in Reference~\cite{vandaele2024optimal} is optimally minimizing the number of Hadamard gates and internal Hadamard gates in parametrized Clifford circuit with non-repeated parameters.

\section{Pauli rotation merging}\label{sec:merging}

In this section, we present efficient algorithms for optimizing the number of non-Clifford gates based on the Pauli rotation merging approach presented in Section~\ref{sub:pauli_rotation}.
For completeness, we first describe the algorithm presented in Reference~\cite{zhang2019optimizing} for optimizing the number of $T$ gates in Clifford$+T$ circuits.
We will use the term \texttt{TMerge} to denote this procedure.
The input circuit must first be put into the form of Equation~\ref{eq:sequence_prod}: a sequence of $\pi/4$ Pauli rotations associated with the Pauli products $P_1, \ldots, P_m$ and a final Clifford operator denoted by $C_f$.
Then, the algorithm proceeds as follows.
Let $U$ and $C$ be equal to the identity operator.
For $i$ going from $1$ to $m$, do the following:
\begin{enumerate}
    \item Add the Pauli rotation $C^\dag R_{P_i}(\pi/4) C$ to $U$, i.e.\ $U \leftarrow C^\dag R_{P_i}(\pi/4) C U$.
    \item Iterate over all other Pauli rotations $R_Q(\pi/4)$ in $U$ until one of the conditions is true:
    \begin{enumerate}
        \item If $Q$ anticommutes with $C^\dag P_i C$.
        \item If $Q = C^\dag P_i C$ then remove $Q$ and $C^\dag P_i C$ from $U$ and $C \leftarrow C R_{P_i}(\pi/2)$.
        \item If $Q = -C^\dag P_i C$ then remove $Q$ and $C^\dag P_i C$ from $U$.
    \end{enumerate}
\end{enumerate}
At the end, the optimized circuit is obtained by performing the synthesis of the sequence of $\pi/4$ Pauli rotations constituting $U$ and the Clifford operator $C_f C$.
Extracting the sequence of Pauli rotations and the final Clifford circuit implemented by the input circuit can be done in $\mathcal{O}(nM)$ where $n$ is the number of qubits and $M$ is the number of gates in the input circuit.
The number of commutativity checks performed by this algorithm in the step 2(a) is $\mathcal{O}(m^2)$ in the worst case.
Each commutativity check requires $\mathcal{O}(n)$ operations, therefore this algorithm has a complexity of $\mathcal{O}(nM + nm^2)$.

In this section we show how the number of commutativity checks can be greatly reduced in some cases.
In Section~\ref{sub:lowering_commutativity_checks}, we present an algorithm that efficiently solves the Pauli rotation merging problem and constructs the associated optimized quantum circuit with a complexity of $\mathcal{O}(nM + nhm)$ where $h$ is the optimal of internal Hadamard gates required to implement the initial sequence of Pauli rotations.
In Section~\ref{sub:extension_t_gates}, we present a variant of the algorithm for optimizing the number of $T$ gates in Clifford$+T$ quantum circuit.

\subsection{Lowering the number of commutativity checks}\label{sub:lowering_commutativity_checks}

In the same way as in Reference~\cite{vandaele2024optimal}, we encode a sequence of $m$ Pauli products into a block matrix $\mathcal{S} = \begin{bmatrix} \mathcal{Z} \\ \mathcal{X} \end{bmatrix}$ of size $2n \times m$ where $n$ is the number of qubits, the submatrix $\mathcal{Z}$ corresponds to the first $n$ rows of $\mathcal{S}$ and the submatrix $\mathcal{X}$ corresponds to the last $n$ rows of $\mathcal{S}$.
The value $({\mathcal{Z}}_{i,j}, \mathcal{X}_{i,j})$ represents the $i$th component of the $j$th Pauli product encoded by $\mathcal{S}$, such that the values $(0, 0), (0, 1), (1, 1)$ and $(1, 0)$ are corresponding to the Pauli matrices $I, X, Y$ and $Z$ respectively.
We use the notation $\mathcal{S}_{:,i}$ to refer to the $(i+1)$th column of the matrix $\mathcal{S}$ and the notation $\mathcal{S}_{:,:i}$ to refer to the submatrix formed by the first $(i+1)$th columns of $\mathcal{S}$.
We will say that $S_{:,i}$ commutes (or anticommutes) with $S_{:,j}$ if the Pauli products that they are respectively encoding are commuting (or anticommuting).
The commutativity matrix associated with a sequence of Pauli products encoded in $\mathcal{S}$ will be denoted $A^{(\mathcal{S})}$.
For convenience we will drop the superscript $(\mathcal{S})$ from $A$ when it is clear from the context that $A$ is associated with $\mathcal{S}$.

A rank vector $\bs v \in \mathbb{F}^m_2$ associated with a commutativity matrix $A$ satisfies $\lvert \bs v \rvert = \rank(A)$ and is defined as follows.

\begin{definition}[Rank vector]
    The rank vector $\bs v \in \mathbb{F}^m_2$ associated with a commutativity matrix $A$ of size $m\times m$ is defined as follows:
    $$
    v_i = 
    \begin{cases}
        0 &\text{if $\rank(A_{:,:i}) = \rank(A_{:,:i-1})$},\\ 
        1 &\text{otherwise},\\ 
    \end{cases}
    $$
    with $v_0 = 0$.
\end{definition}

Let $C$ be a Clifford$+R_Z$ circuit which implements a sequence of Pauli rotations, as described by Equation~\ref{eq:sequence_prod}, encoded by $\mathcal{S}$.
We say that $\bs v$ is the rank vector of $C$ if $\bs v$ is the rank vector associated with the commutativity matrix $A^{(\mathcal{S})}$.
The rank vector $\bs v$ of a circuit $C$ can be computed by Algorithm~\ref{alg:rank_vector}.
The algorithm starts by calling the \texttt{InternalHOpt} procedure presented in Reference~\cite{vandaele2024optimal}.
It has been demonstrated that the \texttt{InternalHOpt} procedure produces a circuit $C'$ with a minimal number of internal Hadamard gates such that the sequences of Pauli rotations and the final Clifford operators of $C$ and $C'$, as given by Equation~\ref{eq:sequence_prod}, are identical.
Also, the number of internal Hadamard gates in $C'$ is equal to the rank of $A^{(\mathcal{S})}$, which corresponds to the Hamming weight of the rank vector $\bs v$ associated with $\mathcal{S}$, where $\mathcal{S}$ is encoding the Pauli products of the sequence of Pauli rotations implemented by $C'$.
Furthermore, if the last $m-i-1$, such that $i$ satisfies $0 \leq i < m$, non-Clifford $R_Z$ gates of $C'$ are removed, then $C'$ implements the sequence of Pauli rotations encoded by $S_{:,:i}$ up to a final Clifford circuit with an optimal number of internal Hadamard gates equal to $\rank(A_{:,:i})$.
We can then deduce that $v_i = 1$ if and only if there is a Hadamard gate between the $i$th and $(i+1)$th non-Clifford $R_Z$ gate of $C'$.

The \texttt{InternalHOpt} procedure has a complexity of $\mathcal{O}(nM + n^2h)$ where $n$ is the number of qubits, $M$ is the number of gates in the input circuit and $h$ is the number of internal Hadamard gates in the produced circuit and satisfies $h < m$ where $m$ is the number of non-Clifford $R_Z$ gates in the input circuit.
The circuit produced by the \texttt{InternalHOpt} procedure contains $\mathcal{O}(nm + n^2)$ gates, therefore the loop in Algorithm~\ref{alg:rank_vector} can be executed in $\mathcal{O}(nm + n^2)$ operations.
Thus, the overall complexity of Algorithm~\ref{alg:rank_vector} is $\mathcal{O}(nM + n^2h)$.

\begin{algorithm}[t]
    \caption{Computes the rank vector associated with a Clifford$+R_Z$ circuit}
    \label{alg:rank_vector}
	\SetAlgoLined
	\SetArgSty{textnormal}
	\SetKwFunction{proc}{RankVector}
	\SetKwInput{KwInput}{Input}
	\SetKwInput{KwOutput}{Output}
    \KwInput{A Clifford$+R_Z$ circuit $C$ containing $m$ non-Clifford $R_Z$ gates.}
    \KwOutput{The rank vector $\bs v$ associated with $C$.}
	\SetKwProg{Fn}{procedure}{}{}
    \Fn{\proc{$C$}}{
        $C' \leftarrow \texttt{InternalHOpt}(C)$ \\
        $\bs v \leftarrow$ vector of size $m$ filled with $0$ \\
        \For{$i \gets 1$ \KwTo $m-1$}{
            \If{there is a Hadamard gate between the $i$th and $(i+1)$th $R_Z$ gate of $C'$}{
                $v_i \leftarrow 1$
            }
        }
        \Return $\bs v$
	}
\end{algorithm}

The following theorem characterizes how the information encoded in the rank vector is valuable for solving the Pauli rotation merging problem (Problem~\ref{pb:rotation_merging}) efficiently.

\begin{theorem}\label{thm:rank_vector}
    Let $\mathcal{S}$ be a sequence of $m$ Pauli products, let $A$ be its commutativity matrix, let $\bs v$ be the rank vector associated with $A$, and let $i$ and $j$ be integers satisfying $0 \leq i < j < m$.
    If $\mathcal{S}_{:,i}$ commutes with $\mathcal{S}_{:,k}$ for all integer $k$ satisfying $i < k \leq j$ and $v_k = 1$, then $\mathcal{S}_{:,i}$ commutes with $\mathcal{S}_{:,k}$ for all $k$ satisfying $i < k \leq j$.
\end{theorem}

\begin{proof}
    For all $k$ satisfying $i < k \leq j$ and $v_k = 1$, if $\mathcal{S}_{:,i}$ commutes with $\mathcal{S}_{:,k}$  then we have $A_{i, k} = 0$.
    Let $k$ be the smallest value satisfying $i < k \leq j$ and $v_k = 0$, then the following equation
    \begin{equation}\label{eq:sum_l}
        \bigoplus_{\ell \in L} A_{:,\ell} = A_{:,k}
    \end{equation}
    is satisfied for some $L \subseteq \{\ell \mid \ell < k\}$.
    We have $A_{i, \ell} = 0$ by definition if $\ell \leq i$,  and if $i < \ell < k$ then $v_\ell = 1$ and so $A_{i, \ell} = 0$ as demonstrated previously.
    Therefore, for all $\ell \in L$ we have $A_{i, \ell} = 0$ and so by Equation~\ref{eq:sum_l} we necessarily have $A_{i, k} = 0$ which means that $\mathcal{S}_{:,i}$ commutes with $\mathcal{S}_{:,k}$.
    It follows that $\mathcal{S}_{:,i}$ commutes with $\mathcal{S}_{:,k}$ for all $k$ satisfying $i < k \leq j$.
\end{proof}

\begin{algorithm}
    \caption{Pauli rotation merging algorithm}
    \label{alg:bb_merge}
	\SetAlgoLined
	\SetArgSty{textnormal}
	\SetKwFunction{proc}{BBMerge}
	\SetKwInput{KwInput}{Input}
	\SetKwInput{KwOutput}{Output}
    \KwInput{A Clifford$+R_Z$ circuit $C$ containing $m$ non-Clifford $R_Z$ gates.}
    \KwOutput{An optimized Clifford$+R_Z$ circuit equivalent to $C$.}
	\SetKwProg{Fn}{procedure}{}{}
    \Fn{\proc{$C$}}{
        $\bs v \leftarrow \texttt{RankVector}(C)$ \\
        $\bs r \leftarrow $ new empty vector \\
        $\mathcal{T} \leftarrow$ new tableau \\
        $\mathcal{S} \leftarrow$ new empty sequence of Pauli products \\
        $\mathcal{D} \leftarrow$ new empty associative array \\
        $t \leftarrow 0$ \\
        \ForEach{gate $G \in C$}{
            \uIf{$G$ is Clifford}{
                Prepend $G^\dag$ to $\mathcal{T}$ \\
            }
            \uElseIf{$G$ is a non-Clifford $R_Z(\theta)$ gate acting on qubit $i$}{
                $r_t \leftarrow \theta$ \\
                $P \leftarrow$ $i$th stabilizer generator of $\mathcal{T}$ \\
                Append $P$ to $\mathcal{S}$ \\
                Append $t$ to $\mathcal{D}[P]$ \\ 
                \uIf{the list $\mathcal{D}[P]$ contains a value not equal to $t$}{
                    $merge \leftarrow \texttt{true}$ \\
                    $j \leftarrow$ last value of the list $\mathcal{D}[P]$ not equal to $t$ \\
                    \For{$k \gets j$ \KwTo $t - 1$}{
                        \uIf{$v_k = 1$ \textbf{and} $P$ anticommutes with $\mathcal{S}_{:,k}$}{
                            $merge \leftarrow \texttt{false}$ \\
                            \texttt{break} \\
                        }
                    }
                    \uIf{$merge$ is \texttt{true}}{
                        $r_t \leftarrow r_j + r_t$ \\
                        $r_j \leftarrow 0$; \\
                        Remove $j$ from $\mathcal{D}[P]$ \\
                        \uIf{$r_t \equiv 0 \pmod{\frac{\pi}{2}}$}{
                            Prepend the $R_Z^\dag(r_t)$ gate acting on qubit $i$ to $\mathcal{T}$ \\
                            Remove $t$ from $\mathcal{D}[P]$ \\
                        }
                    }
                }
                $t \leftarrow t + 1$ \\
            }
        }
        \Return $C$ where the $i$th non-Clifford $R_Z(\theta)$ gate is replaced by $R_Z(r_i)$ for all $i$ \\
	}
\end{algorithm}

Consider the algorithm whose pseudocode is presented in Algorithm~\ref{alg:bb_merge}, and which takes a Clifford$+R_Z$ circuit $C$ as input.
The vector $\bs r$ encodes the angles of the sequence of Pauli rotations associated with $C$ such that the $i$th Pauli rotation have an angle equal to $r_i$.
A tableau $\mathcal{T}$ is a matrix composed of $4n^2 + 2n$ bits encoding the operation performed by a Clifford operator in $\mathcal{C}_n$~\cite{aaronson2004improved}.
The matrix $\mathcal{S}$ is encoding a sequence of Pauli products, as previously defined.
The variable $\mathcal{D}$ is an associative array which maps a Pauli product to a list of integers, we will use $\mathcal{D}[P]$ to refer to the list of integers associated with the Pauli product $P$.
The variable $t$ is a counter for the number of non-Clifford $R_Z$ gates that have been processed so far by the algorithm.

The main loop of the algorithm iterates over the gates of the input circuit $C$.
If the gate $G$ is Clifford then $G^\dag$ is prepended to the tableau $\mathcal{T}$, so that $\mathcal{T}$ is tracking the inverse of the circuit processed so far without the non-Clifford $R_Z$ gates.
Otherwise, if the gate $G$ is not Clifford then it must be a non-Clifford $R_Z(\theta)$ gate acting on some qubit $i$.
The Pauli rotation associated with this $R_Z(\theta)$ gate is $\tilde{C}^\dag R_Z(\theta) \tilde{C}$ where $\tilde{C}$ is the Clifford operator obtained from the input circuit $C$ by truncating it after this $R_Z(\theta)$ gate and by removing all its non-Clifford gates.
The Clifford operator $\tilde{C}^\dag$ is encoded by the tableau $\mathcal{T}$, and therefore the Pauli product $P$ associated with the Pauli rotation $R_P(\theta) = \tilde{C}^\dag R_Z(\theta) \tilde{C}$ is given by the $i$th stabilizer generator of $\mathcal{T}$.
The algorithm then determines whether or not there is a potential Pauli rotation which can be merged with $R_P(\theta)$ by checking if the list $\mathcal{D}[P]$ contains the index associated with another Pauli rotation.
If this is the case, then the corresponding index $j$ is retrieved from the list $\mathcal{D}[P]$.
The value $j$ corresponds to the $j$th Pauli rotation of $C$ when it is put into the form of Equation~\ref{eq:sequence_prod}, which is encoded by $\mathcal{S}_{:,j}$.
We can rely on Theorem~\ref{thm:rank_vector} to determine whether or not the Pauli rotations associated with $\mathcal{S}_{:,j}$ and $P$ can be merged.
If the condition of Theorem~\ref{thm:rank_vector} is not satisfied, then the Pauli rotations will not be merged, and $t$ will remain on the list $\mathcal{D}[P]$ for a potential merge with a Pauli rotation that has not yet been processed.
Otherwise, if the condition is satisfied then the two Pauli rotations can be merged, the angle of the $j$th Pauli rotation encoded by $r_j$ is set to $0$, and the angle of the Pauli rotation associated with $P$ which is encoded by $r_t$ is set to $r_j + r_t$.
Then $j$ must be removed from $\mathcal{D}[P]$ as the $j$th Pauli rotation has been merged, and the variable $t$ is updated.
In the case where $r_j + r_t$ is a multiple of $\pi/2$, then $t$ must be removed from $\mathcal{D}[P]$ as well because it is not associated with a non-Clifford Pauli rotation anymore.
When all gates within $C$ have been processed, the optimized circuit is obtained by replacing the $i$th non-Clifford $R_Z(\theta_i)$ gate from $C$ with $R_Z(r_i)$ for all $i$.\\

\noindent\textbf{Complexity analysis.}
The rank vector $\bs v$ can be computed by Algorithm~\ref{alg:rank_vector} which have a complexity of $\mathcal{O}(nM + n^2h)$ where $M$ is the number of gates in $C$, $n$ is the number of qubits and $h$ is the minimal number of internal Hadamard gates required to implement the sequence of Pauli rotations associated with $C$.
The main loop of Algorithm~\ref{alg:bb_merge} performs $M$ iterations.
Prepending a gate to the tableau $\mathcal{T}$ can be done in $\mathcal{O}(n)$ operations.
The second condition of the main loop (if $G$ is a non-Clifford $R_Z(\theta)$ gate) is evaluated to true $m$ times where $m$ is the number of non-Clifford $R_Z(\theta)$ gates in $C$.
Each operation of the associative array $\mathcal{D}$ can be done with a complexity of $\mathcal{O}(n)$.
Verifying the condition of Theorem~\ref{thm:rank_vector} induces a cost of $\mathcal{O}(nh)$ because $\lvert \bs v \rvert = h$, and so at most $\mathcal{O}(h)$ commutativity checks are performed and a commutativity check can be done in $\mathcal{O}(n)$ operations.
The loop on line 19 of Algorithm~\ref{alg:bb_merge} performs at most $m$ iterations and is executed $\mathcal{O}(m)$ times, leading to a complexity of $\mathcal{O}(m^2)$.
However this can be reduced to a complexity of $\mathcal{O}(hm)$ by using a vector storing the $h$ indices for which $v_i = 1$ instead on using $\bs v$.
All the operations performed when two Pauli rotations are merged can be done with a complexity of $\mathcal{O}(n)$.
Thus, the overall complexity of Algorithm~\ref{alg:bb_merge} is $\mathcal{O}(nM + n^2h + nhm)$, which corresponds to a complexity of $\mathcal{O}(nM + nhm)$ in the typical case where $n < m$.

This algorithm have a better complexity than the $\mathcal{O}(nM + nm^2)$ complexity of the \texttt{TMerge} procedure of Reference~\cite{zhang2019optimizing} because $h < m$.
However, in the worst case we have $h=m-1$, therefore this algorithm and the \texttt{TMerge} procedure have the same worst-case complexity of $\mathcal{O}(nM + nm^2)$.
Nonetheless, for some circuits $h$ is much smaller than $m$, which makes Algorithm~\ref{alg:bb_merge} significantly faster than the \texttt{TMerge} algorithm.
However, although this algorithm leads to the same $T$-count reduction as the \texttt{TMerge} algorithm for most Clifford$+T$ circuits evaluated in the benchmarks of Section~\ref{sec:bench}, it is not always the case for every circuit.
A simple example where Algorithm~\ref{alg:bb_merge}, called the \texttt{BBMerge} procedure, does not find all the Pauli rotations that can be merged is given in Figure~\ref{fig:suboptimality_example}.
This suboptimality is due to the fact that merging two $\pi/4$ Pauli rotations generates a Clifford $\pi/2$ Pauli rotation which modifies the sequence of Pauli rotations.
The rank vector $\bs v$ should be updated to take into account these modifications of the sequence of Pauli rotations.
However, this is not done by Algorithm~\ref{alg:bb_merge} as updating the rank vector everytime two Pauli rotations are merged would induce a much more important worst-case complexity.
As such, the \texttt{BBMerge} procedure finds all possible reductions in the number of non-Clifford $R_Z$ gates using the Pauli rotation merging approach when the angles are black boxes, but it can fail when merging Pauli rotations generates Clifford Pauli rotations.
In the next subsection we propose some modifications to our approach in order to take advantage of the reduced number of commutativity checks of Algorithm~\ref{alg:bb_merge} while obtaining the same $T$-count reduction as the \texttt{TMerge} procedure.
In Section~\ref{sec:opt_parametrized}, we prove that the \texttt{BBMerge} algorithm is optimal when executed on a parametrized Clifford circuit where the parameters are black boxes.

\begin{figure}[t]
    \centering
    \begin{subfigure}{1\textwidth} \centering
    \begin{quantikz}[column sep=0.5cm, row sep=0.35cm]
        & \gate{T} & \gate{H} & \gate{S} & \gate{T} & \gate{T} & \gate{H} & \gate{T} & \qw \\
    \end{quantikz}\vspace{-0.3cm}
    \caption{\centering Initial circuit.}
    \end{subfigure}\vspace{0.5cm}

    \begin{subfigure}{0.49\textwidth} \centering
    \begin{quantikz}[column sep=0.5cm, row sep=0.35cm]
        & \gate{T} & \gate{H} & \gate{S} & \gate{S} & \gate{H} & \gate{T} & \qw \\
    \end{quantikz}\vspace{-0.3cm}
    \caption{\centering Circuit resulting from the execution of the \texttt{BBMerge} procedure.}
    \end{subfigure}
    \begin{subfigure}{0.49\textwidth} \centering
    \begin{quantikz}[column sep=0.5cm, row sep=0.35cm]
        & \gate{H} & \gate{S} & \gate{S} & \gate{H} & \qw \\
    \end{quantikz}\vspace{-0.3cm}
    \caption{\centering Circuit resulting from the execution of the \texttt{FastTMerge} procedure.}
    \end{subfigure}

    \caption{Example of a circuit for which the \texttt{BBMerge} procedure, presented in Algorithm~\ref{alg:bb_merge}, is suboptimal.
    The initial circuit contains $4$ $T$ gates~(a), the \texttt{BBMerge} procedure produces a circuit containing $2$ $T$ gates~(b), and the \texttt{FastTMerge} procedure, presented in Algorithm~\ref{alg:fast_t_merge}, produces a circuit that doesn't contain any $T$ gate~(c).}
    \label{fig:suboptimality_example}
\end{figure}
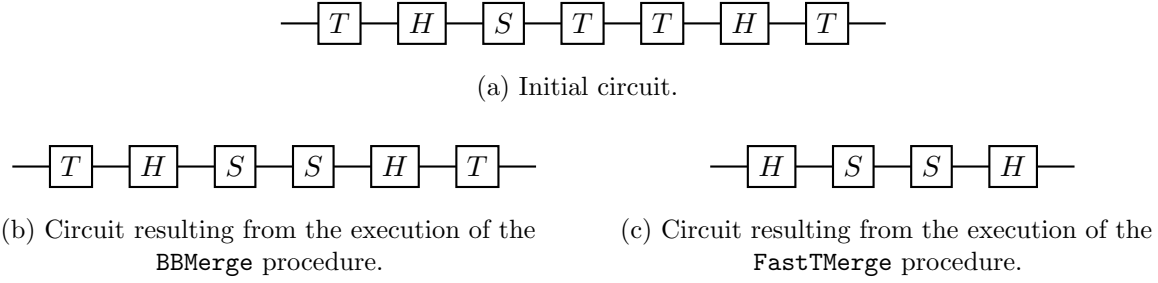

\subsection{Extension for Clifford$+T$ circuits}\label{sub:extension_t_gates}

\begin{algorithm}[tp]
    \caption{Pauli rotation merging algorithm}
    \label{alg:fast_t_merge}
	\SetAlgoLined
	\SetArgSty{textnormal}
	\SetKwFunction{proc}{FastTMerge}
	\SetKwInput{KwInput}{Input}
	\SetKwInput{KwOutput}{Output}
    \KwInput{A Clifford$+R_Z$ circuit $C$ containing $m$ non-Clifford $R_Z$ gates.}
    \KwOutput{An optimized Clifford$+R_Z$ circuit equivalent to $C$.}
	\SetKwProg{Fn}{procedure}{}{}
    \Fn{\proc{$C$}}{
        $\bs v \leftarrow \texttt{RankVector}(C)$ \\
        $\bs w \leftarrow $ copy of $\bs v$ \\
        $\bs r \leftarrow $ new empty vector \\
        $\mathcal{T} \leftarrow$ new tableau \\
        $\mathcal{S} \leftarrow$ new empty sequence of Pauli products \\
        $\mathcal{D} \leftarrow$ new empty associative array\\
        $t \leftarrow 0$\\
        \ForEach{gate $G \in C$}{
            \uIf{$G$ is Clifford}{
                Prepend $G^\dag$ to $\mathcal{T}$ \\
            }
            \uElseIf{$G$ is a non-Clifford $R_Z(\theta)$ gate acting on qubit $i$}{
                $r_t \leftarrow \theta$ \\
                $P \leftarrow$ $i$th stabilizer generator of $\mathcal{T}$ \\
                Append $P$ to $\mathcal{S}$ \\
                Append $t$ to $\mathcal{D}[P]$ \\
                \uIf{the list $\mathcal{D}[P]$ contains a value not equal to $t$}{
                    $merge \leftarrow \texttt{true}$ \\
                    $j \leftarrow$ last value of the list $\mathcal{D}[P]$ not equal to $t$ \\
                    \For{$k \gets j$ \KwTo $t-1$}{
                        \uIf{$v_k = 1$ \textbf{and} $P$ anticommutes with $\mathcal{S}_{:,k}$}{
                            \uIf{$r_k \equiv 0 \pmod{\frac{\pi}{2}}$} {
                                \For{$\ell \gets k+1$ \KwTo $t-1$}{
                                    \uIf{$w_\ell = 1$, $r_\ell \not \equiv 0 \pmod{\frac{\pi}{2}}$ \textbf{and} $P$ anticommutes with $\mathcal{S}_{:,\ell}$} {
                                        $merge \leftarrow \texttt{false}$ \\
                                        \texttt{break} \\
                                    }
                                }
                            }
                            \uElse{
                                $merge \leftarrow \texttt{false}$ \\
                            }
                            \texttt{break} \\
                        }
                    }
                    \uIf{$merge$ is \texttt{true}}{
                        \uIf{$v_j = 1$}{
                            \lFor{$k \gets j+1$ \KwTo $t-1$}{$w_k \leftarrow 1$}
                            $w_j = 0$ \\
                        }
                        $r_t \leftarrow r_t + r_j$ \\
                        $r_j \leftarrow 0$ \\
                        Remove $j$ from $\mathcal{D}[P]$ \\
                        \uIf{$r_t \equiv 0 \pmod{\frac{\pi}{2}}$}{
                            Prepend the $R_Z^\dag(r_t)$ gate acting on qubit $i$ to $\mathcal{T}$ \\
                            Remove $t$ from $\mathcal{D}[P]$ \\
                        }
                    }
                }
                $t \leftarrow t + 1$ \\
            }
        }
        \Return $C$ where the $i$th non-Clifford $R_Z(\theta)$ gate is replaced by $R_Z(r_i)$ for all $i$ \\
	}
\end{algorithm}

In this section, we present a modified version of Algorithm~\ref{alg:bb_merge}, whose pseudocode is given in Algorithm~\ref{alg:fast_t_merge}.
This algorithm achieves the same $T$-count reduction as the \texttt{TMerge} procedure while taking advantage of Theorem~\ref{thm:rank_vector} to lower the number of commutativity checks.
Algorithm~\ref{alg:fast_t_merge} uses the same data structures as Algorithm~\ref{alg:bb_merge} but with an additional vector $\bs w$ which is initialized as a copy of the rank vector $\bs v$.
Let $\mathcal{S}$ be the sequence of Pauli products constructed by Algorithm~\ref{alg:bb_merge} (and by Algorithm~\ref{alg:fast_t_merge} analogously).
As previously, we will use $\mathcal{S}_{:,i}$ to refer to the $i$th Pauli products encoded by the sequence $\mathcal{S}$.
As shown by the example of Figure~\ref{fig:suboptimality_example}, Algorithm~\ref{alg:bb_merge} will not merge two Pauli rotations associated with $\mathcal{S}_{:,i}$ and $\mathcal{S}_{:,j}$ such that $i < j$ and $\mathcal{S}_{:,i} = \mathcal{S}_{:,j}$ when there exists an integer $k$ satisfying $i < k < j$, $v_k = 1$ and $r_k \equiv 0 \pmod{\pi/2}$.
Yet, if $r_k \equiv 0 \pmod{\pi/2}$ then the $k$th Pauli rotation has already been merged with another Pauli rotation to form a Clifford Pauli rotation, and should not prevent the merge of the Pauli rotations associated with $\mathcal{S}_{:,i}$ and $\mathcal{S}_{:,j}$.
This problem could be solved by updating the rank vector $\bs v$ using Algorithm~\ref{alg:rank_vector}.
However, doing this everytime two Pauli rotations are merged would greatly increase the complexity of the algorithm.
Instead, when two Pauli rotations are merged, Algorithm~\ref{alg:fast_t_merge} updates the vector $\bs w$ such that two Pauli rotations associated with $\mathcal{S}_{:,i}$ and $\mathcal{S}_{:,j}$ where $\mathcal{S}_{:,i} = \mathcal{S}_{:,j}$ and $i < j$ can be merged if and only if $\mathcal{S}_{:,i}$ commutes with $\mathcal{S}_{:,k}$ for all $k$ satisfying $i < k < j$, $w_k = 1$ and $r_k \not \equiv 0 \pmod{\pi/2}$.
The vector $\bs w$ is simply updated by setting $w_k$ to $1$ for all $k$ satisfying $i \leq k \leq j$ when the $i$th Pauli rotation is merged with the $j$th Pauli rotation.
We can summarize this process performed by Algorithm~\ref{alg:fast_t_merge} by the following distinction of cases where $i < j$ and $\mathcal{S}_{:,i} = \mathcal{S}_{:,j}$:
\begin{itemize}
    \item If $\mathcal{S}_{:,i}$ commutes with $\mathcal{S}_{:,k}$ for all $k$ satisfying $i < k < j$ and $v_k = 1$, then, as stated by Theorem~\ref{thm:rank_vector}, the Pauli rotations associated with $\mathcal{S}_{:,i}$ and $\mathcal{S}_{:,j}$ can be merged in the same way as in Algorithm~\ref{alg:bb_merge}.
    \item If $\mathcal{S}_{:,i}$ anticommutes with $\mathcal{S}_{:,k}$ for some $k$ satisfying $i < k < j$, $v_k = 1$ and $r_k \not \equiv 0 \pmod{\pi/2}$, then the Pauli rotations associated with $\mathcal{S}_{:,i}$ and $\mathcal{S}_{:,j}$ cannot be merged.
    \item If $\mathcal{S}_{:,i}$ anticommutes with $\mathcal{S}_{:,k}$ for some $k$ satisfying $i < k < j$, $v_k = 1$ and $r_k \equiv 0 \pmod{\pi/2}$, then the algorithm will check whether or not $\mathcal{S}_{:,j}$ anticommutes with $\mathcal{S}_{:,\ell}$ for some $\ell$ satisfying $k < \ell < j$, $w_\ell = 1$ and $r_k \not \equiv 0 \pmod{\pi/2}$.
        If this condition is evaluated to true, then the Pauli rotations associated with $\mathcal{S}_{:,i}$ and $\mathcal{S}_{:,j}$ cannot be merged.
\end{itemize}

This algorithm is then finding all Pauli rotations that can be merged as the sequence of Pauli products $\mathcal{S}$ associated with the resulting sequence of Pauli rotations cannot contain $\mathcal{S}_{:,i}$ and $\mathcal{S}_{:,j}$ such that $i < j$, $\mathcal{S}_{:,i} = \mathcal{S}_{:,j}$ and $\mathcal{S}_{:,i}$ commutes with $\mathcal{S}_{:,k}$ for all $k$ satisfying $i < k < j$.\\

\noindent\textbf{Complexity analysis.}
Let $n$ be the number of qubits, $m$ be the number of non-Clifford $R_Z$ gates in the initial circuit and $M$ be the number of gates in the initial circuit.
When compared to Algorithm~\ref{alg:bb_merge}, Algorithm~\ref{alg:fast_t_merge} is composed of an additional loop (in line 23) which is called at most $m$ times and performs $\mathcal{O}(m)$ iterations.
At each iteration a commutativity check requiring $\mathcal{O}(n)$ operations is performed, therefore this loop induces $\mathcal{O}(nm^2)$ operations.
Algorithm~\ref{alg:fast_t_merge} also has another additional loop (in line 34) to update the vector $\bs w$.
This loop induces a complexity of $\mathcal{O}(hm)$ as it performs $\mathcal{O}(m)$ iterations, each iteration performs a constant number of operations and the loop is called $\mathcal{O}(h)$ times, where $h$ is the minimal number of internal Hadamard gates required to implement the sequence of Pauli rotations associated with the input circuit.
Thus, the worst-case complexity of Algorithm~\ref{alg:fast_t_merge} is $\mathcal{O}(nM + n^2h + nm^2)$, or $\mathcal{O}(nM + nm^2)$ in the typical case where $n < m$.

This complexity matches the complexity of the \texttt{TMerge} procedure of Reference~\cite{zhang2019optimizing} and is worse than the $\mathcal{O}(nM + nhm)$ complexity of Algorithm~\ref{alg:bb_merge} as $h < m$.
However, this algorithm can deal with the specific cases not handled by Algorithm~\ref{alg:bb_merge}, such as the one presented in Figure~\ref{fig:suboptimality_example}.
As such, it reduces the number of non-Clifford $R_Z$ gates as much as possible by using the Pauli rotation merging approach while trying to take advantage of Theorem~\ref{thm:rank_vector} to be more efficient by reducing the number of commutativity checks.
We present benchmarks in Section~\ref{sec:bench} to demonstrate the effectiveness of Algorithm~\ref{alg:bb_merge} and Algorithm~\ref{alg:fast_t_merge}.

\section{Optimality in the number of parametrized rotations}\label{sec:opt_parametrized}

In this section, we prove the following theorem, which states the optimality of Algorithm~\ref{alg:bb_merge} for reducing the number of parametrized rotations in a parametrized Clifford circuit with non-repeated parameters.

\begin{theorem}\label{thm:parametrized_opt}
    Let $C_o$ be the circuit produced by Algorithm~\ref{alg:bb_merge} when a parametrized Clifford circuit with non-repeated parameters is given as input.
    Then $C_o$ is a parametrized Clifford circuit with non-repeated, and the number of parametrized rotations in $C_o$ is optimal: there does not exist a parametrized Clifford circuit with non-repeated parameters equivalent to $C_o$ which contains a lower number of parametrized rotations.
\end{theorem}

Note that the commutativity checks perform in Algorithm~\ref{alg:bb_merge} are up to a global phase.
This is important for parametrized Clifford circuit because it allows some optimizations.
For instance, the following parametrized Clifford circuit
\begin{equation}
    \includegraphics[valign=c]{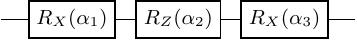}
\end{equation}
where $\alpha_1, \alpha_2, \alpha_3 \in \{0, \pi\}^3$, can be rewritten, up to a global phase, as
\begin{equation}
    \includegraphics[valign=c]{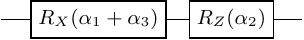}
\end{equation}
which reduces the number of parametrized rotations.
Two Pauli rotations $R_P(\alpha_1)$ and $R_{P'}(\alpha_2)$,  where $\alpha_1 \neq 0 \pmod{2\pi}$ and $\alpha_2 \neq 0 \pmod{2\pi}$, commute, up to a global phase, if and only if $P$ commutes with $P'$ or $\alpha_1 \equiv \alpha_2 \equiv \pi \pmod{2 \pi}$.
Therefore, to execute Algorithm~\ref{alg:bb_merge} on parametrized Clifford circuits, we need to determine if a parametrized rotation gate with angle $\alpha$ satisfies $\alpha = 0 \pmod{\pi}$.
That is the only knowledge required regarding the angles of the parametrized rotation gates.

Theorem~\ref{thm:parametrized_opt} is proven by the following lemma, which states that the number of parametrized rotations in the parametrized Clifford circuit with non-repeated parameters produced by Algorithm~\ref{lem:opt_parametrized_rotations} is optimal, and that all optimal circuits share the same associated sequence of Pauli rotations up to some commutative relations.

\begin{lemma}\label{lem:opt_parametrized_rotations}
    Let $C_o$ be a parametrized Clifford circuit with non-repeated parameters produced by Algorithm~\ref{alg:bb_merge}.
    Let $C'$ be a parametrized Clifford circuit with non-repeated parameters equivalent to $C_o$ and which has an optimal number of parametrized rotations.
    And let $R_{P_1}(f_1(\bs \alpha)), \ldots, R_{P_m}(f_m(\bs \alpha))$ and $R_{P'_1}(f'_1(\bs \alpha)), \ldots, R_{P'_{m'}}(f'_{m'}(\bs \alpha))$ be sequences of Pauli rotations associated with $C_o$ and $C'$ respectively.
    Then $m=m'$ and there exists a permutation $\sigma$ satisfying $R_{P_i}(f_i(\bs \alpha)) = R_{P'_{\sigma(i)}}(f'_{\sigma(i)}(\bs \alpha))$ for all $i$ and such that $R_{P'_{\sigma(i)}}(f'_{\sigma(i)}(\bs \alpha))$ commutes, up to a global phase, with $R_{P'_j}(f'_j(\bs \alpha))$ for all $\bs \alpha$ and for all integer $j$ satisfying $\sigma(i) \leq j \leq i$ and $i \leq j \leq \sigma(i)$.
\end{lemma}

\begin{proof}
    We start the proof of Lemma~\ref{lem:opt_parametrized_rotations} by showing that, for all $\bs \alpha$, $R_{P_1}(f_1(\bs \alpha)) = R_{P'_i}(f'_i(\bs \alpha))$ for some integer $i$, and that $R_{P'_i}(f'_i(\bs \alpha))$ commutes, up to a global phase, with $R_{P'_j}(f'_j(\bs \alpha))$ for all $j$ satisfying $1 \leq j \leq i$.
    Our proof will rely on the following equality, derived from Equation~\ref{eq:sequence_prod_parametrized}:
    \begin{equation}\label{eq:pauli_rotation_sequences_equality}
        \prod_{k=1}^m R_{P_k}(f_k(\bs \alpha)) = \prod_{k=1}^{m'} R_{P'_k}(f'_k(\bs \alpha))
    \end{equation}
    which holds up to a global phase for all $\bs \alpha$ because the parametrized circuits $C_o$ and $C'$ are equivalent.
    Let $S$ be a minimal support of $f_1$ and let $i$ be an integer such that $f'_i(p_S(\bs \alpha)) \neq 0$ for some $\bs \alpha$.
    Note that such a function $f'_i$ necessarily exists, otherwise, because $S$ is a minimal support of $f_1$ and all the $f$ functions have disjoint sets of useful parameters, we would have
    \begin{equation}
        \prod_{k=1}^m R_{P_k}(f_k(p_S(\bs \alpha))) = R_{P_1}(f_1(p_S(\bs \alpha))) \neq I
    \end{equation}
    and
    \begin{equation}
        \prod_{k=1}^{m'} R_{P'_k}(f'_k(p_S(\bs \alpha))) = I
    \end{equation}
    for some $\bs \alpha$, which contradicts Equation~\ref{eq:pauli_rotation_sequences_equality}.

    We first prove that
    \begin{equation} \label{eq:i_minimal_support}
        \prod_{k=1}^{m'} R_{P'_k}(f'_k(p_S(\bs \alpha))) = R_{P'_i}(f'_i(p_S(\bs \alpha)))
    \end{equation}
    for all $\bs \alpha$.
    Let's assume that there exist an integer $j$ and a vector $\bs \alpha$ such that $j \neq i$ and $f'_j(p_S(\bs \alpha)) \neq 0$.
    The set $S$ is not a minimal support of $f'_j$ because there exists $\bs \alpha$ such that $f'_i(p_S(\bs \alpha)) \neq 0$ and $f'_i$ and $f'_j$ have disjoint sets of useful parameters.
    It entails that there exists a minimal support $T$ of $f'_j$ such that $T \subset S$.
    Moreover, $S$ is a minimal support of $f_1$, and so $f_k(p_T(\bs \alpha)) = 0$ for all $\bs \alpha$ and for all $k$ satisfying $1 \leq k \leq m$.
    Then, there exists an $\bs \alpha$ such that
    \begin{equation}
        \prod_{k=1}^m R_{P_k}(f_k(p_T(\bs \alpha))) = I
    \end{equation}
    and
    \begin{equation}
        \prod_{k=1}^{m'} R_{P'_k}(f'_k(p_T(\bs \alpha))) = R_{P'_j}(f'_j(p_T(\bs \alpha)) \neq I
    \end{equation}
    which contradicts Equation~\ref{eq:pauli_rotation_sequences_equality}.
    Therefore, by contradiction, there do not exist $j$ and $\bs \alpha$ such that $j \neq i$ and $f'_j(p_S(\bs \alpha)) \neq 0$.
    Thus, Equation~\ref{eq:i_minimal_support} holds, implying that, based on Equation~\ref{eq:pauli_rotation_sequences_equality}, we have
    \begin{equation}\label{eq:opt_parametrized_lem_8}
        R_{P_1}(f_1(p_S(\bs \alpha))) = R_{P'_i}(f'_i(p_S(\bs \alpha)))
    \end{equation}
    for all $\bs \alpha$, and so $P_1$ = $P'_i$ and
    \begin{equation}
        f_1(p_S(\bs \alpha)) = f'_i(p_S(\bs \alpha))
    \end{equation}
    for all $\bs \alpha$.

    We now prove that $R_{P'_i}(f'_i(\bs \alpha))$ commutes, up to a global phase, with $R_{P'_j}(f'_j(\bs \alpha))$ for all $\bs \alpha$ and for all integer $j$ satisfying $1 \leq j \leq i$.
    Let $j$ be an integer satisfying $1 \leq j < i$ and such that $R_{P'_j}(f'_j(\bs \alpha))$ anticommutes with $R_{P'_i}(f'_i(\bs \alpha))$ for some $\bs \alpha$.
    Let $T$ be a minimal support of $f'_j$ and let's assume that there exists $\bs \alpha$ such that $f_1(p_T(\bs \alpha)) \neq 0$.
    Then there would exist a minimal support $U$ of $f_1$ such that $U \subseteq T$.
    Based on Equation~\ref{eq:pauli_rotation_sequences_equality}, we then have
    \begin{equation}
        R_{P_1}(f_1(p_U(\bs \alpha))) = R_{P'_j}(f'_j(p_U(\bs \alpha)))
    \end{equation}
    for all $\bs \alpha$, which entails $P_1 = P'_j$.
    However, we also have $P_1 = P'_i$, which implies $P'_i = P'_j$.
    This contradicts our assumption that $R_{P'_j}(f'_j(\bs \alpha))$ anticommutes with $R_{P'_i}(f'_i(\bs \alpha))$ for some $\bs \alpha$.
    Therefore, by contradiction, we have $f_1(p_T(\bs \alpha)) = 0$ for all $\bs \alpha$.
    It follows that, based on Equation~\ref{eq:pauli_rotation_sequences_equality}, we have
    \begin{equation}\label{eq:opt_parametrized_lem_4}
        \prod_{k=1}^m R_{P_k}(f_k(p_T(\bs \alpha))) = \prod_{k=2}^m R_{P_k}(f_k(p_T(\bs \alpha))) = \prod_{k=1}^{m'} R_{P'_k}(f'_k(p_T(\bs \alpha))) = R_{P'_j}(f'_j(p_T(\bs \alpha)))
    \end{equation}
    for all $\bs \alpha$, because $T \subseteq \param(f'_j)$ and all the $f'$ functions have disjoint sets of useful parameters.
    It entails
    \begin{equation}
        T \subseteq \bigcup_{k=2}^m \param(f_k)
    \end{equation}
    and so 
    \begin{equation}
        T \cap \param(f_1) = \emptyset
    \end{equation}
    because all the $f$ functions have disjoint sets of useful parameters.
    It implies that 
    \begin{equation}\label{eq:opt_parametrized_lem_0}
        R_{P_1}(f_1(p_{S\cup T}(\bs \alpha))) = R_{P_1}(f_1(p_S(\bs \alpha)))
    \end{equation}
    for all $\bs \alpha$.
    Moreover, we have
    \begin{equation}\label{eq:opt_parametrized_lem_1}
        R_{P_k}(f_k(p_{S\cup T}(\bs \alpha))) = R_{P_k}(f_k(p_T(\bs \alpha)))
    \end{equation}
    for all integer $k$ satisfying $1 < k \leq m$, because $S \subseteq \param(f_1)$ and $f_1$ and $f_k$ have disjoint sets of useful parameters.
    Based on Equations~\ref{eq:opt_parametrized_lem_4},~\ref{eq:opt_parametrized_lem_0} and~\ref{eq:opt_parametrized_lem_1}, we have
    \begin{equation}\label{eq:opt_parametrized_lem_5}
        \begin{aligned}
        \prod_{k=1}^{m} R_{P_k}(f_k(p_{S \cup T}(\bs \alpha))) &= R_{P_1}(f_1(p_S(\bs \alpha))) \prod_{k=2}^{m} R_{P_k}(f_k(p_T(\bs \alpha))) \\
        &= R_{P_1}(f_1(p_S(\bs \alpha))) R_{P'_j}(f'_j(p_T(\bs \alpha)))
        \end{aligned}
    \end{equation}
    for all $\bs \alpha$.
    Furthermore, we have
    \begin{equation}\label{eq:opt_parametrized_lem_6}
        \prod_{k=1}^{m'} R_{P'_k}(f'_k(p_{S \cup T}(\bs \alpha))) = R_{P'_j}(f'_j(p_T(\bs \alpha))) R_{P'_i}(f'_i(p_S(\bs \alpha)))
    \end{equation}
    for all $\bs \alpha$, because $j<i$, $T \subseteq \param(f'_j)$, $S \subseteq \param(f'_i)$ and all $f'$ functions have disjoint sets of useful parameters.
    According to Equation~\ref{eq:pauli_rotation_sequences_equality}, Equations~\ref{eq:opt_parametrized_lem_5} and~\ref{eq:opt_parametrized_lem_6} should be equal:
    \begin{equation}\label{eq:opt_parametrized_lem_7}
        R_{P_1}(f_1(p_S(\bs \alpha))) R_{P'_j}(f'_j(p_T(\bs \alpha))) = R_{P'_j}(f'_j(p_T(\bs \alpha))) R_{P'_i}(f'_i(p_S(\bs \alpha)))
    \end{equation}
    for all $\bs \alpha$.
    Using Equation~\ref{eq:opt_parametrized_lem_8}, Equation~\ref{eq:opt_parametrized_lem_7} can be rewritten as
    \begin{equation}\label{eq:opt_parametrized_lem_9}
        R_{P'_i}(f'_i(p_S(\bs \alpha))) R_{P'_j}(f'_j(p_T(\bs \alpha))) = R_{P'_j}(f'_j(p_T(\bs \alpha))) R_{P'_i}(f'_i(p_S(\bs \alpha)))
    \end{equation}
    for all $\bs \alpha$.
    Equation~\ref{eq:opt_parametrized_lem_9} indicates that $R_{P'_i}(f'_i(\bs \alpha))$ commutes with $R_{P'_j}(f'_j(\bs \alpha))$ for all $\bs \alpha$.
    This contradicts our assumption that $R_{P'_j}(f'_j(\bs \alpha))$ anticommutes with $R_{P'_i}(f'_i(\bs \alpha))$ for some $\bs \alpha$.
    Therefore, by contradiction, $R_{P'_i}(f'_i(\bs \alpha))$ commutes, up to a global phase, with $R_{P'_j}(f'_j(\bs \alpha))$ for all $\bs \alpha$ and for all integer $j$ satisfying $1 \leq j \leq i$.

    We now prove that 
    \begin{equation}
        R_{P_1}(f_1(\bs \alpha)) = R_{P'_i}(f'_i(\bs \alpha))
    \end{equation}
    for all $\bs \alpha$.
    To do so, we will prove that
    \begin{equation}
        \param(f_1) = \param(f'_i)
    \end{equation}
    by showing that $\param(f'_i) \subseteq \param(f_1)$ and $\param(f_1) \subseteq \param(f'_i)$.

    Let's assume that there exists an integer $j$ and an $\bs \alpha$ such that $j \neq i$ and $f'_j(p_{\param(f_1)}(\bs \alpha)) \neq 0$.
    Let $T$ be a minimal support of $f'_j$ such that $T \subseteq \param(f_1)$.
    Based on Equation~\ref{eq:pauli_rotation_sequences_equality}, and because all functions $f$ and $f'$ have disjoint sets of useful parameters, we have 
    \begin{equation}\label{eq:opt_parametrized_lem_18}
        R_{P_1}(f_1(p_T(\bs \alpha))) = R_{P'_j}(f'_j(p_T(\bs \alpha)))
    \end{equation}
    for all $\bs \alpha$, which implies that $P_1 = P'_j = P'_i$ because $R_{P'_j}(f'_j(p_T(\bs \alpha))) \neq I$ for some $\bs \alpha$.
    Let $u$ be an integer between $i$ and $j$ such that the Pauli rotation $R_{P'_u}(f'_u(\bs \alpha))$ anticommutes with $R_{P'_j}(f'_j(\bs \alpha))$ for some $\bs \alpha$.
    Notice that $u$ necessarily exists as otherwise the Pauli rotations $R_{P'_i}(f'_i(\bs \alpha))$ and $R_{P'_j}(f'_j(\bs \alpha))$ could be merged into a single Pauli rotation because $P'_i = P'_j$, which would contradict the fact that $C'$ has an optimal number of parametrized rotations.
    Also, $i, u$ and $j$ must satisfy $i < u < j$ because we proved that $R_{P'_i}(f'_i(\bs \alpha))$ commutes with $R_{P'_k}(f'_k(\bs \alpha))$ for all $\bs \alpha$ and $k$ satisfying $1 \leq k < i$.
    Let $U$ be a minimal support of $f'_u$, and let's assume that there exists $\bs \alpha$ such that $f_1(p_U(\bs \alpha)) \neq 0$.
    Then there would exists a minimal support $V$ of $f_1$ such that $V \subseteq U$.
    Based on Equation~\ref{eq:pauli_rotation_sequences_equality}, we then have
    \begin{equation}\label{eq:opt_parametrized_lem_10}
        R_{P_1}(f_1(p_V(\bs \alpha))) = R_{P'_u}(f'_u(p_V(\bs \alpha)))
    \end{equation}
    for all $\bs \alpha$.
    Because there exists $\bs \alpha$ such that $f_1(p_V(\bs \alpha)) \neq 0$, Equation~\ref{eq:opt_parametrized_lem_10} implies that $P_1 = P'_u$.
    However, we also have $P_1 = P'_j$ and $P'_j \neq P_u$ because there exists an $\bs \alpha$ such that $R_{P'_j}(f'_j(\bs \alpha))$ anticommutes with $R_{P'_u}(f'_u(\bs \alpha))$.
    Therefore, by contradiction, we have $f_1(p_U(\bs \alpha)) = 0$ for all $\bs \alpha$.
    It follows that, based on Equation~\ref{eq:pauli_rotation_sequences_equality}, we have
    \begin{equation}\label{eq:opt_parametrized_lem_14}
        \prod_{k=1}^m R_{P_k}(f_k(p_U(\bs \alpha))) = \prod_{k=2}^m R_{P_k}(f_k(p_U(\bs \alpha))) = \prod_{k=1}^{m'} R_{P'_k}(f'_k(p_U(\bs \alpha))) = R_{P'_u}(f'_u(p_U(\bs \alpha)))
    \end{equation}
    for all $\bs \alpha$, because $U \subseteq \param(f'_u)$ and all the $f'$ functions have disjoint sets of useful parameters.
    It entails
    \begin{equation}
        U \subseteq \bigcup_{k=2}^m \param(f_k)
    \end{equation}
    and so 
    \begin{equation}
        U \cap \param(f_1) = \emptyset
    \end{equation}
    because all the $f$ functions have disjoint sets of useful parameters.
    It implies that 
    \begin{equation}\label{eq:opt_parametrized_lem_11}
        R_{P_1}(f_1(p_{T\cup U}(\bs \alpha))) = R_{P_1}(f_1(p_T(\bs \alpha)))
    \end{equation}
    for all $\bs \alpha$.
    Moreover, we have
    \begin{equation}\label{eq:opt_parametrized_lem_12}
        R_{P_k}(f_k(p_{T\cup U}(\bs \alpha))) = R_{P_k}(f_k(p_U(\bs \alpha)))
    \end{equation}
    for all integer $k$ satisfying $1 < k \leq m$, because $T \subseteq \param(f_1)$ and $f_1$ and $f_k$ have disjoint sets of useful parameters.
    Based on Equations~\ref{eq:opt_parametrized_lem_14},~\ref{eq:opt_parametrized_lem_11} and~\ref{eq:opt_parametrized_lem_12}, we have
    \begin{equation}\label{eq:opt_parametrized_lem_13}
        \begin{aligned}
            \prod_{k=1}^{m} R_{P_k}(f_k(p_{T \cup U}(\bs \alpha))) &= R_{P_1}(f_1(p_T(\bs \alpha))) \prod_{k=2}^{m} R_{P_k}(f_k(p_U(\bs \alpha))) \\
            &= R_{P_1}(f_1(p_T(\bs \alpha))) R_{P'_u}(f'_u(p_U(\bs \alpha)))
        \end{aligned}
    \end{equation}
    for all $\bs \alpha$.
    Furthermore, we have
    \begin{equation}\label{eq:opt_parametrized_lem_15}
        \prod_{k=1}^{m'} R_{P'_k}(f'_k(p_{T \cup U}(\bs \alpha))) = R_{P'_u}(f'_u(p_U(\bs \alpha))) R_{P'_j}(f'_j(p_T(\bs \alpha)))
    \end{equation}
    for all $\bs \alpha$, because $u<j$, $U \subseteq \param(f'_u)$, $T \subseteq \param(f'_j)$ and all $f'$ functions have disjoint sets of useful parameters.
    According to Equation~\ref{eq:pauli_rotation_sequences_equality}, Equations~\ref{eq:opt_parametrized_lem_13} and~\ref{eq:opt_parametrized_lem_15} should be equal:
    \begin{equation}\label{eq:opt_parametrized_lem_16}
        R_{P_1}(f_1(p_T(\bs \alpha))) R_{P'_u}(f'_u(p_U(\bs \alpha))) = R_{P'_u}(f'_u(p_U(\bs \alpha))) R_{P'_j}(f'_j(p_T(\bs \alpha)))
    \end{equation}
    for all $\bs \alpha$.
    Using Equation~\ref{eq:opt_parametrized_lem_18}, Equation~\ref{eq:opt_parametrized_lem_16} can be rewritten as
    \begin{equation}\label{eq:opt_parametrized_lem_17}
        R_{P'_j}(f'_j(p_T(\bs \alpha))) R_{P'_u}(f'_u(p_U(\bs \alpha))) = R_{P'_u}(f'_u(p_U(\bs \alpha))) R_{P'_j}(f'_j(p_T(\bs \alpha)))
    \end{equation}
    for all $\bs \alpha$.
    Equation~\ref{eq:opt_parametrized_lem_17} indicates that $R_{P'_u}(f'_u(\bs \alpha))$ commutes with $R_{P'_j}(f'_j(\bs \alpha))$ for all $\bs \alpha$.
    This contradicts the fact that $R_{P'_u}(f'_u(\bs \alpha))$ anticommutes with $R_{P'_j}(f'_j(\bs \alpha))$ for some $\bs \alpha$.
    Therefore, by contradiction, there does not exist an integer $j$ and an $\bs \alpha$ such that $j \neq i$ and $f'_j(p_{\param(f_1)}(\bs \alpha)) \neq 0$.
    Thus, based on Equation~\ref{eq:pauli_rotation_sequences_equality},
    \begin{equation}
        R_{P_1}(f_1(p_{\param(f_1)}(\bs \alpha))) = R_{P'_i}(f'_i(p_{\param(f_1)}(\bs \alpha)))
    \end{equation}
    for all $\bs \alpha$, which implies $\param(f'_i) \subseteq \param(f_1)$.

    Let's assume that there exists an integer $j$ and an $\bs \alpha$ such that $1 < j$ and $f_j(p_{\param(f'_i)}(\bs \alpha)) \neq 0$.
    Let $T$ be a minimal support of $f_j$ such that $T \subseteq \param(f'_i)$.
    Based on Equation~\ref{eq:pauli_rotation_sequences_equality}, and because all functions $f$ and $f'$ have disjoint sets of useful parameters, we have 
    \begin{equation}\label{eq:opt_parametrized_lem_18_b}
        R_{P_j}(f_j(p_T(\bs \alpha))) = R_{P'_i}(f'_i(p_T(\bs \alpha)))
    \end{equation}
    for all $\bs \alpha$, which implies that $P'_i = P_j = P_1$ because $R_{P_j}(f_j(p_T(\bs \alpha))) \neq I$ for some $\bs \alpha$.
    Let $u$ be an integer satisfying $1 < u < j$ such that the Pauli rotation $R_{P_u}(f_u(\bs \alpha))$ anticommutes with $R_{P_j}(f_j(\bs \alpha))$ for some $\bs \alpha$.
    Notice that $u$ necessarily exists as otherwise the Pauli rotations $R_{P_1}(f_1(\bs \alpha))$ and $R_{P_j}(f_j(\bs \alpha))$ could be merged into a single Pauli rotation because $P_i = P_j$, which would contradict the fact that $C_o$ is a circuit produced by Algorithm~\ref{alg:bb_merge}.
    Let $U$ be a minimal support of $f_u$, and let's assume that there exists $\bs \alpha$ such that $f'_i(p_U(\bs \alpha)) \neq 0$.
    Then there would exists a minimal support $V$ of $f'_i$ such that $V \subseteq U$.
    Based on Equation~\ref{eq:pauli_rotation_sequences_equality}, we then have
    \begin{equation}\label{eq:opt_parametrized_lem_10_b}
        R_{P_u}(f_u(p_V(\bs \alpha))) = R_{P'_i}(f'_i(p_V(\bs \alpha)))
    \end{equation}
    for all $\bs \alpha$.
    Because there exists $\bs \alpha$ such that $f'_i(p_V(\bs \alpha)) \neq 0$, Equation~\ref{eq:opt_parametrized_lem_10_b} implies that $P'_i = P_u$.
    However, we also have $P'_i = P_j$ and $P_j \neq P_u$ because there exists an $\bs \alpha$ such that $R_{P_j}(f_j(\bs \alpha))$ anticommutes with $R_{P_u}(f_u(\bs \alpha))$.
    Therefore, by contradiction, we have $f'_i(p_U(\bs \alpha)) = 0$ for all $\bs \alpha$.
    It follows that, based on Equation~\ref{eq:pauli_rotation_sequences_equality}, we have
    \begin{equation}\label{eq:opt_parametrized_lem_14_b}
        \prod_{k=1}^{m'} R_{P'_k}(f'_k(p_U(\bs \alpha))) = \prod_{k=1, k\neq i}^{m'} R_{P'_k}(f'_k(p_U(\bs \alpha))) = \prod_{k=1}^{m} R_{P_k}(f_k(p_U(\bs \alpha))) = R_{P_u}(f_u(p_U(\bs \alpha)))
    \end{equation}
    for all $\bs \alpha$, because $U \subseteq \param(f_u)$ and all the $f$ functions have disjoint sets of useful parameters.
    It entails
    \begin{equation}
        U \subseteq \bigcup_{k=2}^m \param(f'_k)
    \end{equation}
    and so 
    \begin{equation}
        U \cap \param(f'_i) = \emptyset
    \end{equation}
    because all the $f$ functions have disjoint sets of useful parameters.
    It implies that 
    \begin{equation}\label{eq:opt_parametrized_lem_11_b}
        R_{P'_i}(f'_i(p_{T\cup U}(\bs \alpha))) = R_{P'_i}(f'_i(p_T(\bs \alpha)))
    \end{equation}
    for all $\bs \alpha$.
    Moreover, we have
    \begin{equation}\label{eq:opt_parametrized_lem_12_b}
        R_{P'_k}(f'_k(p_{T\cup U}(\bs \alpha))) = R_{P'_k}(f'_k(p_U(\bs \alpha)))
    \end{equation}
    for all integer $k$ satisfying $1 \leq k \leq m$, $k \neq i$, because $T \subseteq \param(f'_i)$ and $f'_i$ and $f'_k$ have disjoint sets of useful parameters.
    Based on Equations~\ref{eq:opt_parametrized_lem_14_b},~\ref{eq:opt_parametrized_lem_11_b} and~\ref{eq:opt_parametrized_lem_12_b}, we have
    \begin{equation}\label{eq:opt_parametrized_lem_13_b}
        \begin{aligned}
            \prod_{k=1}^{m'} R_{P'_k}(f'_k(p_{T \cup U}(\bs \alpha))) &= R_{P'_i}(f'_i(p_T(\bs \alpha))) \prod_{k=1, k \neq i}^{m'} R_{P'_k}(f'_k(p_U(\bs \alpha))) \\
            &= R_{P'_i}(f'_i(p_T(\bs \alpha))) R_{P_u}(f_u(p_U(\bs \alpha)))
        \end{aligned}
    \end{equation}
    for all $\bs \alpha$ because we proved that $R_{P'_i}(f'_i(\bs \alpha))$ commutes with $R_{P'_k}(f'_k(\bs \alpha))$ for all $\bs \alpha$ and all integer $k$ satisfying $1 \leq k \leq i$.
    Furthermore, we have
    \begin{equation}\label{eq:opt_parametrized_lem_15_b}
        \prod_{k=1}^{m} R_{P_k}(f_k(p_{T \cup U}(\bs \alpha))) = R_{P_u}(f_u(p_U(\bs \alpha))) R_{P_j}(f_j(p_T(\bs \alpha)))
    \end{equation}
    for all $\bs \alpha$, because $u<j$, $U \subseteq \param(f_u)$, $T \subseteq \param(f_j)$ and all $f$ functions have disjoint sets of useful parameters.
    According to Equation~\ref{eq:pauli_rotation_sequences_equality}, Equations~\ref{eq:opt_parametrized_lem_13_b} and~\ref{eq:opt_parametrized_lem_15_b} should be equal:
    \begin{equation}\label{eq:opt_parametrized_lem_16_b}
        R_{P'_i}(f'_i(p_T(\bs \alpha))) R_{P_u}(f_u(p_U(\bs \alpha))) = R_{P_u}(f_u(p_U(\bs \alpha))) R_{P_j}(f_j(p_T(\bs \alpha)))
    \end{equation}
    for all $\bs \alpha$.
    Using Equation~\ref{eq:opt_parametrized_lem_18_b}, Equation~\ref{eq:opt_parametrized_lem_16_b} can be rewritten as
    \begin{equation}\label{eq:opt_parametrized_lem_17_b}
        R_{P_j}(f_j(p_T(\bs \alpha))) R_{P_u}(f_u(p_U(\bs \alpha))) = R_{P_u}(f_u(p_U(\bs \alpha))) R_{P_j}(f_j(p_T(\bs \alpha)))
    \end{equation}
    for all $\bs \alpha$.
    Equation~\ref{eq:opt_parametrized_lem_17_b} indicates that $R_{P_j}(f_j(\bs \alpha))$ commutes with $R_{P_u}(f_u(\bs \alpha))$ for all $\bs \alpha$.
    This contradicts the fact that $R_{P_u}(f_u(\bs \alpha))$ anticommutes with $R_{P_j}(f_j(\bs \alpha))$ for some $\bs \alpha$.
    Therefore, by contradiction, there does not exist an integer $j$ and an $\bs \alpha$ such that $1 < j$ and $f_j(p_{\param(f'_i)}(\bs \alpha)) \neq 0$.
    Thus, based on Equation~\ref{eq:pauli_rotation_sequences_equality},
    \begin{equation}
        R_{P_1}(f_1(p_{\param(f'_i)}(\bs \alpha))) = R_{P'_i}(f'_i(p_{\param(f'_i)}\bs \alpha)))
    \end{equation}
    for all $\bs \alpha$, which implies $\param(f_1) \subseteq \param(f'_i)$.

    We proved that $\param(f'_i) \subseteq \param(f_1)$ and $\param(f_1) \subseteq \param(f'_i)$, which implies $\param(f_1) = \param(f'_i)$.
    Based on Equation~\ref{eq:pauli_rotation_sequences_equality} and because all $f$ and $f'$ functions have disjoint sets of useful parameters, it follows that
    \begin{equation}
        R_{P_1}(f_1(\bs \alpha)) = R_{P'_i}(f'_i(\bs \alpha))
    \end{equation}
    for all $\bs \alpha$.

    We can now conclude the proof of Lemma~\ref{lem:opt_parametrized_rotations} by induction.
    Let $k$ be an integer satisfying $0 < k < m$.
    Let's assume that there exists an injective function $\sigma$ satisfying $R_{P_i}(f_i(\bs \alpha)) = R_{P'_{\sigma(i)}}(f'_{\sigma(i)}(\bs \alpha))$ for all $i$ satisfying $1 \leq i \leq k$ and such that $R_{P'_{\sigma(i)}}(f'_{\sigma(i)}(\bs \alpha))$ commutes with $R_{P'_j}(f'_j(\bs \alpha))$ for all $\bs \alpha$ and all integer $j$ satisfying $\sigma(i) \leq j \leq i$ and $i \leq j \leq \sigma(i)$.
    Note that we already proved this in the case where $k=1$.
    Let $\tilde{C_o}$ and $\tilde{C}'$ be the parametrized Clifford circuits with non-repeated parameters obtained from $C_o$ and $C'$, respectively, by removing the parametrized rotation gates associated with the Pauli rotations $R_{P_i}(f_i(\bs \alpha))$ and $R_{P'_{\sigma(i)}}(f'_{\sigma(i)}(\bs \alpha))$ for all $i$ satisfying $1 \leq i \leq k$.
    The parametrized Clifford circuits $\tilde{C_o}$ and $\tilde{C}'$ are equivalent because $R_{P_i}(f_i(\bs \alpha)) = R_{P'_{\sigma(i)}}(f'_{\sigma(i)}(\bs \alpha))$ for all $i$ satisfying $1 \leq i \leq k$, and $P'_{\sigma(i)}$ commutes with $P'_j$ for all $j$ satisfying $\sigma(i) \leq j \leq i$ and $i \leq j \leq \sigma(i)$.
    Also, because $C_o$ is produced by Algorithm~\ref{alg:bb_merge}, there does not exist two integers $i, j$ satisfying $k < i < j \leq m$ and $P_i = P_j$, and such that $R_{P_j}(f_j(\bs \alpha))$ commutes with $R_{P_\ell}(f(_\ell(\bs \alpha))$ for all integer $\ell$ satisfying $i \leq \ell \leq j$.
    Therefore, the parametrized Clifford circuit with non-repeated parameters produced by Algorithm~\ref{alg:bb_merge} when $\tilde{C_o}$ is given as input is identical to $\tilde{C_o}$.
    Moreover, and $\tilde{C}'$ has an optimal number of parametrized rotations because $C'$ has an optimal number of parametrized rotations.
    Thus, the circuits $\tilde{C_o}$ and $\tilde{C}'$ are satisfying the conditions of Lemma~\ref{lem:opt_parametrized_rotations}.
    Then, as proven above, $R_{P_{k+1}}(f_{k+1}(\bs \alpha)) = R_{P'_{i}}(f'_{i}(\bs \alpha))$ for all $\bs \alpha$ and some integer $i$ where $R_{P'_{i}}(f'_{i}(\bs \alpha))$ belongs to the sequence of Pauli rotations associated with $\tilde{C}'$, and $R_{P'_i}(f'_i(\bs \alpha))$ commutes, up to a global phase, with $R_{P'_j}(f'_j(\bs \alpha))$ for all $j$ satisfying $i \leq j \leq k+1$ and $k+1 \leq j \leq k+1$.
    Hence, by induction, $m = m'$ and there exists a permutation $\sigma$ satisfying $R_{P_i}(f_i(\bs \alpha)) = R_{P'_{\sigma(i)}}(f'_{\sigma(i)}(\bs \alpha))$ for all $i$ and such that $R_{P'_{\sigma(i)}}(f'_{\sigma(i)}(\bs \alpha))$ commutes, up to a global phase, with $R_{P'_j}(f'_j(\bs \alpha))$ for all $\bs \alpha$ and all integer $j$ satisfying $\sigma(i) \leq j \leq i$ and $i \leq j \leq \sigma(i)$.
\end{proof}

A direct consequence of Lemma~\ref{lem:opt_parametrized_rotations} is that the relations between the parameters of two equivalent  parametrized Clifford circuits with non-repeated parameters are corresponding to additive relations.
In Reference~\cite{van2024optimal}, it was proven that parameter transformations restricted to analytic functions are actually equal to affine functions.
Lemma~\ref{lem:opt_parametrized_rotations} extends this result for the general case, thereby showing that affine functions are truly all we need to describe the parameter transformations.

Theorem~\ref{thm:parametrized_opt} is not exclusive to Algorithm~\ref{alg:bb_merge}; the result also holds for any algorithm which finds all the parametrized rotations that can be merged together.
Among these algorithms, Algorithm~\ref{alg:bb_merge} is the one having the best complexity.
As discussed in Section~\ref{sec:merging}, the algorithm of Reference~\cite{zhang2019optimizing} has a complexity of $\mathcal{O}(nM + nm^2)$ where $n$ is the number of qubits, $M$ is the number of gates and $m$ is the number of parametrized rotations.
This is worse than the $\mathcal{O}(nM + nhm)$ complexity of Algorithm~\ref{alg:bb_merge}, where $h$ satisfies $h < m$ and is the minimal number of internal Hadamard gates.
In Reference~\cite{kissinger2020reducing}, another approach is proposed with a complexity of $\mathcal{O}(N^3)$, where $N$ is the number of spiders in the ZX-diagram associated with the initial circuit.
This complexity is also worse than the one of Algorithm~\ref{alg:bb_merge} because $N$ satisfies $h < m \leq N$ and $n \leq N$.
The performances of these algorithms are compared in Section~\ref{sec:bench}.

\section{Optimality in the number of Hadamard gates and internal Hadamard gates}\label{sec:opt_parametrized_h}

In this section, we present similar results for the number of Hadamard gates and internal Hadamard gates.
The number of internal Hadamard gates in a parametrized quantum circuit corresponds to the number of Hadamard between the first and last parametrized rotation gate of the circuit.
In Reference~\cite{vandaele2024optimal}, an algorithm is presented for implementing the sequence of Pauli rotations and the following Clifford operator of Equation~\ref{eq:sequence_prod} with a minimal number of Hadamard gates and internal Hadamard gates over the $\{\mathrm{CNOT}, H, S, X, R_Z\}$ gate set.
In this section, we will base our results on this algorithm and we will also consider that the Clifford gates in parametrized Clifford circuits are from the $\{\mathrm{CNOT}, H, S, X\}$ gate set.
We prove the following theorem, which states the optimality of the algorithm presented in Reference~\cite{vandaele2024optimal} for reducing the number of Hadamard and internal Hadamard gates in a parametrized Clifford circuit with non-repeated parameters

\begin{theorem}\label{thm:parametrized_h_opt}
    Let $C_o$ be the parametrized Clifford circuit with non-repeated parameters produced by the algorithm presented in Reference~\cite{vandaele2024optimal} when a parametrized Clifford circuit with non-repeated parameters is given as input.
    Then, he number of Hadamard gates and the number of internal Hadamard gates in $C_o$ is optimal: there does not exist a parametrized Clifford circuit with non-repeated parameters equivalent to $C_o$ which contains a lower number of Hadamard gates or a lower number of internal Hadamard gates.
\end{theorem}

We will rely on the following lemma to prove Theorem~\ref{thm:parametrized_h_opt}, which states that applying Algorithm~\ref{alg:bb_merge} on a parametrized Clifford circuit do not change the rank of its associated commutativity matrix.

\begin{lemma}\label{lem:bb_merge_rank}
    Let $C$ be a parametrized Clifford circuits with non-repeated parameters and let $C_o$ be the parametrized Clifford circuit with non-repeated parameters produced by Algorithm~\ref{alg:bb_merge} when $C$ is given as input.
    Let $A$ and $A_o$ be the commutativity matrices associated with $C$ and $C_o$ respectively. 
    Then $\rank(A) = \rank(A_o)$.
\end{lemma}

\begin{proof}
    Let $R_{P_1}(f_1(\bs \alpha)), \ldots, R_{P_m}(f_m(\bs \alpha))$ be a sequence of Pauli rotations associated with $C$.
    Let $i$ be an integer satisfying $1 \leq i < m$ and such that $P_i$ commutes with $P_{i+1}$.
    Let $\tilde{A}$ be the commutativity matrix associated with the sequence of Pauli rotations where the positions of the $i$th and $(i+1)$th Pauli rotations have been swapped:
    $$R_{P_1}(f_1(\bs \alpha)), \ldots, R_{P_{i-1}}(f_{i-1}(\bs \alpha)), R_{P_{i+1}}(f_{i+1}(\bs \alpha)), R_{P_{i}}(f_{i}(\bs \alpha)), R_{P_{i+2}}(f_{i+2}(\bs \alpha)), \ldots, R_{P_{m}}(f_{m}(\bs \alpha)).$$
    Because $P_i$ commutes with $P_{i+1}$, by definition we have $A_{i, i+1} = \tilde{A}_{i, i+1} = 0$.
    It follows that $A_{j, i} = \tilde{A}_{j, i+1}$, $A_{j, i+1} = \tilde{A}_{j, i}$ and $A_{i, j} = \tilde{A}_{i+1, j}$, $A_{i+1, j} = \tilde{A}_{i, j}$ for all $j$.
    I.e.\ the matrix $\tilde{A}$ can be obtained from the matrix $A$ by swapping its $i$th and $(i+1)$th columns and its $i$th and $(i+1)$th rows.
    Performing these operations on the matrix $A$ do not change its rank, therefore $\rank(A) = \rank(\tilde{A})$.

    Let $i$ be an integer satisfying $1 \leq i < m$ and such that $P_i = P_{i+1}$.
    Let $\tilde{A}$ be the commutativity matrix associated with the sequence of Pauli rotations where the $i$th and $(i+1)$th Pauli rotations have been merged into a single Pauli rotation:
    $$R_{P_1}(f_1(\bs \alpha)), \ldots, R_{P_{i}}(f_i(\bs \alpha) + f_{i+1}(\bs \alpha)), R_{P_{i+2}}(f_{i+2}(\bs \alpha)), \ldots, R_{P_{m}}(f_m(\bs \alpha)).$$
    Because $P_{i} = P_{i+1}$, by definition we have $A_{j, i} = A_{j, i+1} = \tilde{A}_{j, i}$ and $A_{i, j} = A_{i+1, j} = \tilde{A}_{i, j}$ for all $j$.
    I.e.\ the $i$th column of $A$ is equal to the $(i+1)$th column of $A$, the $i$th row of $A$ is equal to the $(i+1)$th row of $A$, and the matrix $\tilde{A}$ can be obtained from the matrix $A$ by removing its $(i+1)$th column and its $(i+1)$th row.
    Performing these operations on the matrix $A$ do not change its rank, therefore $\rank(A) = \rank(\tilde{A})$.

    As stated by Lemma~\ref{lem:opt_parametrized_rotations}, the sequence of Pauli rotations associated with $C_o$ can be obtained from the sequence of Pauli rotations associated with $C$ by repeating any of the two following operations a finite number of times:
    \begin{itemize}
        \item Swap the positions of two adjacent commuting Pauli rotations.
        \item Merge two adjacent Pauli rotations $R_{P_i}(f_i(\bs \alpha))$ and $R_{P_{i+1}}(f_{i+1}(\bs \alpha))$ satisfying $P_i = P_{i+1}$ into a single Pauli rotation $R_{P_i}(f_i(\bs \alpha) + f_{i+1}(\bs \alpha))$.
    \end{itemize}
    We proved that these two operations do not change the rank of the associated commutativity matrix.
    Thus, $\rank(A) = \rank(A_o)$.
\end{proof}

Based on Lemma~\ref{lem:bb_merge_rank}, we can prove the following lemma, which states that the commutativity matrices and the extended commutativity matrices associated with the parametrized Clifford circuits with non-repeated parameters produced by Algorithm~\ref{alg:bb_merge} all have the same rank.

\begin{lemma}\label{lem:h_opt_rank_eq}
    Let $C$ and $C'$ be equivalent  parametrized Clifford circuits with non-repeated parameters.
    Let $A$ and $A'$ be the commutativity matrices associated with $C$ and $C'$ respectively.
    And let $M$ and $M'$ be the extended commutativity matrices associated with $C$ and $C'$ respectively.
    Then $\rank(A) = \rank(A')$ and $\rank(M) = \rank(M')$.
\end{lemma}

\begin{proof}
    Let $C_o$ and $C'_o$ be the parametrized Clifford circuits with non-repeated parameters produced by Algorithm~\ref{alg:bb_merge} when $C$ and $C'$ are given as input.
    And let $A_o$ and $A'_o$ be the commutativity matrices associated with $C_o$ and $C'_o$ respectively.
    Based on Lemma~\ref{lem:bb_merge_rank} we have $\rank(A_o) = \rank(A)$ and $\rank(A'_o) = \rank(A')$.
    Moreover, Lemma~\ref{lem:opt_parametrized_rotations} states that $C_o$ and $C'_o$ have the same number of parametrized rotations, and that the sequence of Pauli rotations associated with $C_o$ can be obtained from the sequence of Pauli rotations associated with $C'_o$ by repeatedly swapping the position of two adjacent commuting Pauli rotations.
    As shown in the proof Lemma~\ref{lem:bb_merge_rank}, swapping the position of two adjacent commuting Pauli rotations do not change the rank of the associated commutativity matrix.
    It follows that $\rank(A_o) = \rank(A'_o)$, which implies that $\rank(A) = \rank(A')$.

    Let $\tilde{C}$ and $\tilde{C}'$ be circuits obtained from $C$ and $C'$ by inserting a parametrized $R_Z$ gate at the beginning and the end of the circuit on each qubit, such that $\tilde{C}$ and $\tilde{C}'$ are equivalent  parametrized Clifford circuits with non-repeated parameters.
    As proven above, the commutativity matrices associated with $\tilde{C}$ and $\tilde{C}'$ have the same rank because $\tilde{C}$ and $\tilde{C}'$ are equivalent.
    Moreover, the commutativity matrices of $\tilde{C}$ and $\tilde{C}'$ are the extended commutativity matrices of $C$ and $C'$ respectively, denoted $M$ and $M'$.
    Thus, $\rank(M) = \rank(M')$.
\end{proof}

The proof of Theorem~\ref{thm:parametrized_h_opt} can now be formulated by relying on Lemma~\ref{lem:h_opt_rank_eq}.

\begin{proof}[Proof of Theorem~\ref{thm:parametrized_h_opt}]
    Let $A_o$ and $M_o$ be the commutativity matrix and the extended commutativity matrix associated with $C_o$.
    The algorithm of Reference~\cite{vandaele2024optimal} produces a circuit with a number of internal Hadamard gates equal to $\rank(A_o)$ and a number of Hadamard gates equal to $\rank(M_o)$.
    Lemma~\ref{lem:h_opt_rank_eq} states that any parametrized Clifford circuits with non-repeated parameters equivalent to $C_o$ with an associated commutativity matrix $A$ and an associated extended commutativity matrix $M$ satisfies $\rank(A) = \rank(A_o)$ and $\rank(M) = \rank(M_o)$.
    Moreover, in a quantum circuit composed of gates from the set $\{\mathrm{CNOT}, S, H, X, R_Z\}$, the number of internal Hadamard gates and the number of Hadamard gates cannot be lower than the rank of its associated commutativity matrix and the rank of its associated extended commutativity matrix respectively~\cite{vandaele2024optimal}.
    Therefore, there does not exist a parametrized Clifford circuit with non-repeated parameters equivalent to $C_o$ which contains a number of internal Hadamard lower than $\rank(A_o)$ or a number of Hadamard gates lower than $\rank(M_o)$.
\end{proof}

\section{Benchmarks}\label{sec:bench}

\begin{table}
\resizebox{1.0\columnwidth}{!}{
\begin{tabular}{lrrrrrrrrrrrrr} 
        \toprule
         & && \multicolumn{2}{c}{\texttt{PyZX}~\cite{kissinger2020reducing}} && \multicolumn{2}{c}{\texttt{BBMerge}} && \multicolumn{2}{c}{\texttt{FastTMerge}} && \multicolumn{2}{c}{\texttt{TMerge}~\cite{zhang2019optimizing}} \\
        \cmidrule(lr){4-5} \cmidrule(lr){7-8} \cmidrule(lr){10-11} \cmidrule(lr){13-14}
        Circuit & $T$-count && $T$-count & t (s) && $T$-count & t (s) && $T$-count & t (s) && $T$-count & t (s) \\
        \midrule
        Adder$_8$ & 399 && 173 & 0 && \bf 179 & 0 && 173 & 0 && 173 & 0 \\
        Barenco Tof$_3$ & 28 && 16 & 0 && 16 & 0 && 16 & 0 && 16 & 0 \\
        Barenco Tof$_4$ & 56 && 28 & 0 && 28 & 0 && 28 & 0 && 28 & 0 \\
        Barenco Tof$_5$ & 84 && 40 & 0 && 40 & 0 && 40 & 0 && 40 & 0 \\
        Barenco Tof$_{10}$ & 224 && 100 & 0 && 100 & 0 && 100 & 0 && 100 & 0 \\
        CSLA MUX$_3$ & 70 && 62 & 0 && 62 & 0 && 62 & 0 && 62 & 0 \\
        CSUM MUX$_9$ & 196 && 84 & 0 && 84 & 0 && 84 & 0 && 84 & 0 \\
        DEFAULT & 62720 && - & - && 39744 & 6 && 39744 & 6 && 39744 & 8 \\
        GF$(2^4)$ Mult & 112 && 68 & 0 && 68 & 0 && 68 & 0 && 68 & 0 \\
        GF$(2^5)$ Mult & 175 && 115 & 0 && 115 & 0 && 115 & 0 && 115 & 0 \\
        GF$(2^6)$ Mult & 252 && 150 & 0 && 150 & 0 && 150 & 0 && 150 & 0 \\
        GF$(2^7)$ Mult & 343 && 217 & 0 && 217 & 0 && 217 & 0 && 217 & 0 \\
        GF$(2^8)$ Mult & 448 && 264 & 0 && 264 & 0 && 264 & 0 && 264 & 0 \\
        GF$(2^9)$ Mult & 567 && 351 & 0 && 351 & 0 && 351 & 0 && 351 & 0 \\
        GF$(2^{10})$ Mult & 700 &&  410 & 1 && 410 & 0 && 410 & 0 && 410 & 0 \\
        GF$(2^{16})$ Mult & 1792 && 1040 & 4 && 1040 & 0 && 1040 & 0 && 1040 & 0 \\
        GF$(2^{32})$ Mult & 7168 && 4128 & 47 && 4128 & 0 && 4128 & 0 && 4128 & 6 \\
        GF$(2^{64})$ Mult & 28672 && 16448 & 423 && 16448 & 0 && 16448 & 0 && 16448 & 99 \\
        GF$(2^{128})$ Mult & 114688 && 65664 & 8917 && 65664 & 2 && 65664 & 2 && 65664 & 1930 \\
        GF$(2^{256})$ Mult & 458752 && - & - && 262400 & 18 && 262400 & 18 && 262400 & 42304 \\
        GF$(2^{512})$ Mult & 1835008 && - & - && 1048576 & 137 && 1048576 & 138 && - & - \\
        Grover$_5$ & 336 && 166 & 0 && 166 & 0 && 166 & 0 && 166 & 0 \\
        Ham$_{15}$ (high) & 2457 && 1019 & 9 && \bf 1021 & 0 && 1019 & 0 && 1019 & 0 \\
        Ham$_{15}$ (low) & 161 && 97 & 0 && 97 & 0 && 97 & 0 && 97 & 0 \\
        Ham$_{15}$ (med) & 574 && 212 & 2 && \bf 242 & 0 && 212 & 0 && 212 & 0 \\
        HWB$_6$ & 105 && 75 & 0 && 75 & 0 && 75 & 0 && 75 & 0 \\
        HWB$_8$ & 5887 && 3517 & 43 && \bf 3583 & 0 && 3517 & 0 && 3517 & 0 \\
        HWB$_{10}$ & 29939 && 15891 & 910  && \bf 16371 & 0 && 15891 & 0 && 15891 & 1 \\
        HWB$_{11}$ & 84196 && 44500 & 7473 && \bf 45610 & 1 && 44500 & 1 && 44500 & 1 \\
        HWB$_{12}$ & 171465 && 85611 & 52910 && \bf 87035 & 2 && 85611 & 2 && 85611 & 5 \\
        Mod Adder$_{1024}$ & 1995 && 1011 & 3 && 1011 & 0 && 1011 & 0 && 1011 & 0 \\
        Mod Mult$_{55}$ & 49 && 35 & 0 && 35 & 0 && 35 & 0 && 35 & 0 \\
        Mod Red$_{21}$ & 119 && 73 & 0 && 73 & 0 && 73 & 0 && 73 & 0 \\
        Mod5$_4$ & 28 && 8 & 0 && 8 & 0 && 8 & 0 && 8 & 0 \\
        QCLA Adder$_{10}$ & 238 && 162 & 0 && 162 & 0 && 162 & 0 && 162 & 0 \\
        QCLA Com$_7$ & 203 && 95 & 0 && 95 & 0 && 95 & 0 && 95 & 0 \\
        QCLA Mod$_7$ & 413 && 237 & 0 && 237 & 0 && 237 & 0 && 237 & 0 \\
        QFT$_4$ & 69 && 67 & 0 && 67 & 0 && 67 & 0 && 67 & 0 \\
        RC Adder$_6$ & 77 && 47 & 0 && 47 & 0 && 47 & 0 && 47 & 0 \\
        Tof$_3$ & 21 && 15 & 0 && 15 & 0 && 15 & 0 && 15 & 0 \\
        Tof$_4$ & 36 && 23 & 0 && 23 & 0 && 23 & 0 && 23 & 0 \\
        Tof$_5$ & 49 && 31 & 0 && 31 & 0 && 31 & 0 && 31 & 0 \\
        Tof$_{10}$ & 119 && 71 & 0 && 71 & 0 && 71 & 0 && 71 & 0 \\
        VBE Adder$_3$ & 70 && 24 & 0 && 24 & 0 && 24 & 0 && 24 & 0 \\
        \bottomrule
\end{tabular}
}
\caption{Comparison of different procedures for optimizing the number of $T$ gates in Clifford$+T$ circuits.
    The first $T$-count column indicates the number of $T$ gates in the initial circuits.
    The $T$-count after optimization and the execution time in seconds is reported for each procedure.
    A blank entry indicates that the execution couldn’t be carried out in less than a day.
    A $T$-count entry is in bold if the \texttt{BBMerge} procedure didn't achieve the same $T$-count as the other methods.}\label{tab:bench_merging}
\end{table}

In this section we evaluate the performances of the algorithms we presented, and we compare them to state-of-the-art alternatives.
The benchmarks are performed over a set of Clifford$+T$ circuits obtained from References~\cite{amyGithub, reversibleBenchmarks}.
We also include a circuit implementing the block cipher DEFAULT, as given in Reference~\cite{default}.
We compare the performances of the algorithm presented in Reference~\cite{kissinger2020reducing}, the \texttt{BBMerge} procedure (Algorithm~\ref{alg:bb_merge}), the \texttt{FastTMerge} procedure (Algorithm~\ref{alg:fast_t_merge}) and the \texttt{TMerge} procedure of Reference~\cite{zhang2019optimizing}.
The \texttt{BBMerge}, \texttt{FastTMerge} and \texttt{TMerge} procedures were implemented with the Rust programming language, while the algorithm of Reference~\cite{kissinger2020reducing} was implemented in Python in the PyZX library~\cite{kissinger2020Pyzx}.
Our implementations of Algorithm~\ref{alg:bb_merge} and Algorithm~\ref{alg:fast_t_merge} used for the benchmarks are open source~\cite{github}.

The results of our benchmarks are presented in Table~\ref{tab:bench_merging}.
As expected, the \texttt{FastTMerge} procedure and the algorithms presented in Reference~\cite{zhang2019optimizing} and Reference~\cite{kissinger2020reducing} all obtain the same reduction in the number of $T$ gates.
The \texttt{BBMerge} procedure has the lowest execution time but it fails to achieve the same number of $T$ gates as the other algorithms on a few circuits.
This indicates that, for these circuits, if the non-Clifford rotation gates are treated as black boxes, as done by the \texttt{BBMerge} procedure, it is not possible to find all the non-Clifford rotation gates that can be merged into a single rotation gate.
To achieve better optimization, as done by the other algorithms, it is necessary to know the angles of the non-Clifford rotation gates.

We can notice that the \texttt{BBMerge} and \texttt{FastTMerge} procedures have very similar execution times, despite not having the same complexity.
Also, on large circuits, these procedures have much lower execution times than the algorithms presented in References~\cite{zhang2019optimizing, kissinger2020reducing}.
For instance, the ``GF($2^{512}$) Mult" circuit can be optimized in less than 3 minutes by the \texttt{BBMerge} and \texttt{FastTMerge} procedures, whereas the other algorithms could not optimize this circuit whithin a day.
The \texttt{BBMerge} and \texttt{FastTMerge} procedures are particularly efficient on the ``GF($2^{n}$) Mult" circuits because these circuits can be rewritten without any internal Hadamard gates.
Thus, by computing the rank vector, the procedures will detect that all the Pauli rotations associated with the circuit are commuting, and so no commutativity check will have to be performed.
This best-case scenario for the \texttt{BBMerge} and \texttt{FastTMerge} procedures is, on the contrary, the worst-case scenario for the \texttt{TMerge} procedure of Reference~\cite{zhang2019optimizing}.
Indeed, in such cases, the \texttt{TMerge} procedure will have to perform $\mathcal{O}(m^2)$ commutativity checks, where $m$ is the number of $T$ gates in the initial circuit.

The worst-case complexity of the \texttt{BBMerge} procedure is $\mathcal{O}(nM + nhm)$, whereas the worst-case complexity of the \texttt{FastTMerge} procedure is $\mathcal{O}(nM + nm^2)$, where $n$ is the number of qubits, $M$ is the number of gates in the initial circuit, $m$ is the number of $T$ gates in the initial circuit, and $h$ satisfies $h < m$ and is the minimal number of internal Hadamard gates required to implement the sequence of Pauli rotations associated with the initial circuit.
Thus, an open problem is to bridge the complexity gap between the \texttt{BBMerge} procedure and the \texttt{FastTMerge} procedure.
Does there exist an algorithm that achieves the same $T$-count reduction as the \texttt{FastTMerge} procedure but with the same time complexity as the \texttt{BBMerge} procedure?
If such an algorithm exists, then, without increasing the time complexity, it could be coupled with the algorithm presented in Reference~\cite{vandaele2024optimal}, which finds the optimal number of Hadamard gates and internal Hadamard gates in the circuit as long as the associated sequence of Pauli rotations is not altered.
This would result in a particularly efficient procedure for optimizing both the number of $T$ gates and the number of Hadamard gates, which can be used as a pre-processing step to improve the efficiency of more advanced $T$-count optimizers~\cite{heyfron2018efficient, de2020fast, ruiz2024quantum}.

\section{Conclusion}

We presented an algorithm to optimize the number of non-Clifford rotation gates in a quantum circuit.
Our algorithm has a better complexity than other algorithms having the same optimization strategy and benchmarks show that its execution time is much lower.
When executed on a parametrized quantum circuit, we proved that our algorithm produces a circuit containing a minimal number of parametrized rotations under the condition that all constant rotations of the circuit are Clifford rotations and that no parameter appears more than once.
Under the same conditions, we also proved that the Hadamard gate optimization algorithm of Reference~\cite{vandaele2024optimal} produces a circuit containing a minimal number of Hadamard gates and internal Hadamard gates.

It was first shown in Reference~\cite{van2024optimal} that specific knowledge of the phases involved in the circuit is required to do better at removing non-Clifford rotation gates.
Our results corroborate this with some interesting additional insights.
In particular, our results hold for any kind of parameter transformations that do not clone parameters.
Notably, we showed that allowing parameter transformations that are discontinuous does not enable better optimization.
For instance, this imply that the Euler angle transformation is not useful to find an equivalent parametrized Clifford circuit with non-repeated parameters containing a lower number of parametrized rotations.
Also, it is interesting to note that specific knowledge of the phases involved in the circuit is required not only to better optimize the number of non-Clifford rotation gates, but also to better optimize the number of Hadamard gates.
This further reinforces the interdependence between the non-Clifford rotation gate optimization problem and the Hadamard gate optimization problem, as previously highlighted in Reference~\cite{vandaele2024optimal}.

An obvious generalization of our results would be to remove one of the two restrictions required to claim optimality.
The first one is that all constant rotations must be Clifford rotations.
It has been proven that allowing non-Clifford constant rotations leads to an NP-hard problem~\cite{van2024optimal}.
Removing the second restriction would consist in allowing the parameters to appear more than once in the circuit.
This problem is conjecture to be NP-hard~\cite{van2024optimal}.
Due to the projected difficulty of these generalizations, we might consider smaller incremental steps.
For instance, we could allow parameters to be repeated only a small number of times.
Another generalization could be to allow repeated parameters only for parametrized Clifford rotations.
Any progress on these problems would bring us closer to an optimal compilation of parametrized quantum circuits.

\section*{Acknowledgments}
We acknowledge funding from the Plan France 2030 through the projects NISQ2LSQ ANR-22-PETQ-0006 and EPIQ ANR-22-PETQ-007.

\bibliographystyle{unsrt}
\bibliography{ref.bib}

\end{document}